\documentclass[%
reprint,
superscriptaddress,
nofootinbib,
amsmath,amssymb,
aps,longbibliography,
pra,
]{revtex4-2}

\newif\ifarxiv 
\arxivtrue

\usepackage{amsthm,bm,xparse,xcolor,tikz,physics}
\usepgfmodule{decorations}
\usetikzlibrary{decorations,decorations.pathmorphing,patterns}
\usepackage[T1]{fontenc}

\AtBeginDocument{%
    \newwrite\bibnotes
    \def\bibnotesext{Notes.bib}
    \immediate\openout\bibnotes=\jobname\bibnotesext
    \immediate\write\bibnotes{@CONTROL{REVTEX42Control}}
    \immediate\write\bibnotes{@CONTROL{%
    apsrev42Control,author="08",editor="1",pages="1",title="0",year="0"}}
     \if@filesw
     \immediate\write\@auxout{\string\citation{apsrev42Control}}%
    \fi
  }%

\usepackage[nodayofweek]{datetime}  
\usepackage{hyperref}
\usepackage{xcolor}
\usepackage{makecell}
\usepackage[capitalise]{cleveref}
\usepackage{enumerate}   
\definecolor{mylinkcolor}{rgb}{0,0,0.8} 
\hypersetup{unicode=true, %
  bookmarksnumbered=false,bookmarksopen=false,bookmarksopenlevel=1, %
  breaklinks=true,pdfborder={0 0 0},colorlinks=true}%
\hypersetup{%
  anchorcolor=mylinkcolor,citecolor=mylinkcolor, %
  filecolor=mylinkcolor,linkcolor=mylinkcolor, %
  menucolor=mylinkcolor,runcolor=mylinkcolor, %
  urlcolor=mylinkcolor}%

\newtheorem{theorem}{Theorem}
\newtheorem{lemma}{Lemma}
\newtheorem{corollary}{Corollary}

\theoremstyle{definition}

\newtheorem{remark}{Remark}

\usepackage{soul}
\newcommand{\id}{\openone}

\newcommand{\phv}{P_{H/V}}
\newcommand{\pad}{P_{D/A}}
\newcommand{\pdc}{P_{DC}}
\newcommand{\pdcn}{P_{DC}^N}
\newcommand{\pin}{\tilde{\Pi}^N}
\newcommand{\pik}{\tilde{\Pi}^k}
\newcommand{\HN}{\mathcal{H}_N}

\newcommand{\I}{\mathrm{i}}

\begin{document}
\title{Numerical security framework for quantum key distribution with bypass channels}
\author{Lewis Wooltorton}
    \email{lewis.wooltorton@ens-lyon.fr}
    \thanks{Current affiliation: Inria, ENS de Lyon.}
    \affiliation{Department of Mathematics, University of York, Heslington, York, YO10 5DD, United Kingdom}
    \affiliation{Quantum Engineering Centre for Doctoral Training, H. H. Wills Physics Laboratory and Department of Electrical \& Electronic Engineering, University of Bristol, Bristol BS8 1FD, United Kingdom}
    \affiliation{Inria, ENS de Lyon, LIP, 46 Allee d’Italie, 69364 Lyon Cedex 07, France}

\author{Twesh Upadhyaya}
    \affiliation{Joint Center for Quantum Information and Computer Science, University of Maryland and NIST, College Park, MD 20742, USA}
    \affiliation{Department of Physics, University of Maryland, College Park, MD 20742, USA}

\author{Mohsen Razavi}
    \email{m.razavi@leeds.ac.uk}
    \affiliation{School of Electronic and Electrical Engineering, University of Leeds, Leeds LS2 9JT, United Kingdom}

\date{\today}

\begin{abstract}
        
        Satellite based quantum key distribution (QKD) aims to establish secure key exchange over long distances despite significant technological challenges. To alleviate some of these challenges, Ghalaii et al. [PRX Quantum 4, 040320 (2023)] proposed that any airborne eavesdropper up to a certain size can be detected by classical monitoring techniques, limiting the transmission efficiencies of any undetected Eve. This creates a new QKD scenario in which some of the transmitted signal from Alice to Bob bypasses Eve entirely. In this manuscript, we develop a general framework for computing key rates in this ``bypass'' scenario for discrete variable protocols. We first numerically support a conjecture that the performance of BB84 with single photons does not improve under bypass constraints, and go on to find new regimes that do. Specifically, we find improvements when the receiver's detectors have an efficiency mismatch and when BB84 is implemented using weak coherent pulses under certain squashing assumptions. Technically, our framework is realized by including marginal constraints on the source to account for bypass effects, combined with existing numerical approaches for minimizing the key rate and squashing and dimension reduction techniques to handle photonic states of unbounded dimension. 
\end{abstract}
\maketitle

\ifarxiv\section{Introduction}\label{sec:intro}\else\noindent{\it Introduction.|}\fi
Quantum key distribution (QKD) enables two separated parties, Alice and Bob, to share private randomness in the presence of an unbounded quantum adversary~\cite{BB84,Ekert,BBM92,B92,Bruss98,Pirandola_2020}. Security can be established by specifying a model of the underlying source and measurements used by Alice and Bob to generate the raw key. Then, assuming the eavesdropper, Eve, can freely manipulate the quantum signals sent, bounds on the minimum private randomness contained in the raw key can be derived, conditioned on the observed value of certain parameters, and supported by authenticated classical communication between Alice and Bob~\cite{ShorPreskill,DW,Renner}. Further classical post-processing ultimately leads to a final key meeting the required security conditions.

In practice, real Eves may not just be bounded by the constraints of quantum theory. For example, motivated by the current status of quantum technologies, one might consider adding computational assumptions regarding Eve's access to a quantum memory~\cite{vyas2020everlasting}. Or, in the context of satellite based QKD~\cite{Bonato_2009,Moli-Sanchez09,Nauerth2013,Wang2013,Hosseinidehaj19}, a possible solution to long distance quantum communications~\cite{Liao_2017,Ren_2017,Liao_2018}, the so called ``wiretap channel'' has been explored~\cite{Vergoossen19,Pan20,Vazquez-Castro21}. This uses the assumption that for line-of-sight links between a satellite and a ground station, it would be challenging for an airborne eavesdropping object to intercept and resend any signal reliably. Eve is therefore likely positioned on the Earth's surface, outside a ``safe-zone'' surrounding the ground station. Consequently, she can only access a fraction of the signal transmitted from the satellite, leading to key rate improvements.  

Asserting that the underlying channel is a wiretap channel is an unverifiable assumption however. Instead, we would like to focus on the classical limitations of Eve that can be assessed by experimental data. This includes constraints derived from the physical system, or resource limitations, for example. Applying this idea to satellite QKD, it was recently proposed by Ghalaii et al.~\cite{ghalaii2023satellitebased} that Alice and Bob can directly monitor the satellite to ground link via classical means. This will alert the users to eavesdropping objects above a certain size, hence limiting the size of any undetected Eve. Consequently, her collection and resend capabilities are restricted, implying some of the signal sent from Alice could miss Eve entirely. Moreover, this signal might still find its way to Bob, resulting in a new QKD scenario featuring bypass channels. 

The \textit{bypass scenario} is shown in \cref{fig:bypass_gen}, which models such a channel. We consider a lossy channel between Alice and Eve, with transmittance $\eta_{AE}$, followed by another lossy channel between Eve and Bob, with transmittance $\eta_{EB}$. The values of $\eta_{AE}$ and $\eta_{EB}$ are to be bounded following classical monitoring techniques such as LIDAR, and determine Eve's restriction (i.e., the fraction of the signal she receives and can send to Bob). Additionally, signals not collected by Eve may still find their way to Bob's collection device. We model this using a channel between Alice and Bob that bypasses Eve. Importantly, whilst this channel is inaccessible to Eve, it is also unknown to Alice and Bob; they must therefore optimize over all such channels compatible with experimental observations and choose the worst case scenario to bound the key rate.  Aside from its relevance to satellite QKD, analyzing QKD security in the presence of bypass channels is an interesting mathematical problem in its own right, and demonstrates how to include physical limitations as eavesdropping restrictions.  

\begin{figure}[h]
\includegraphics[width=8.4cm]{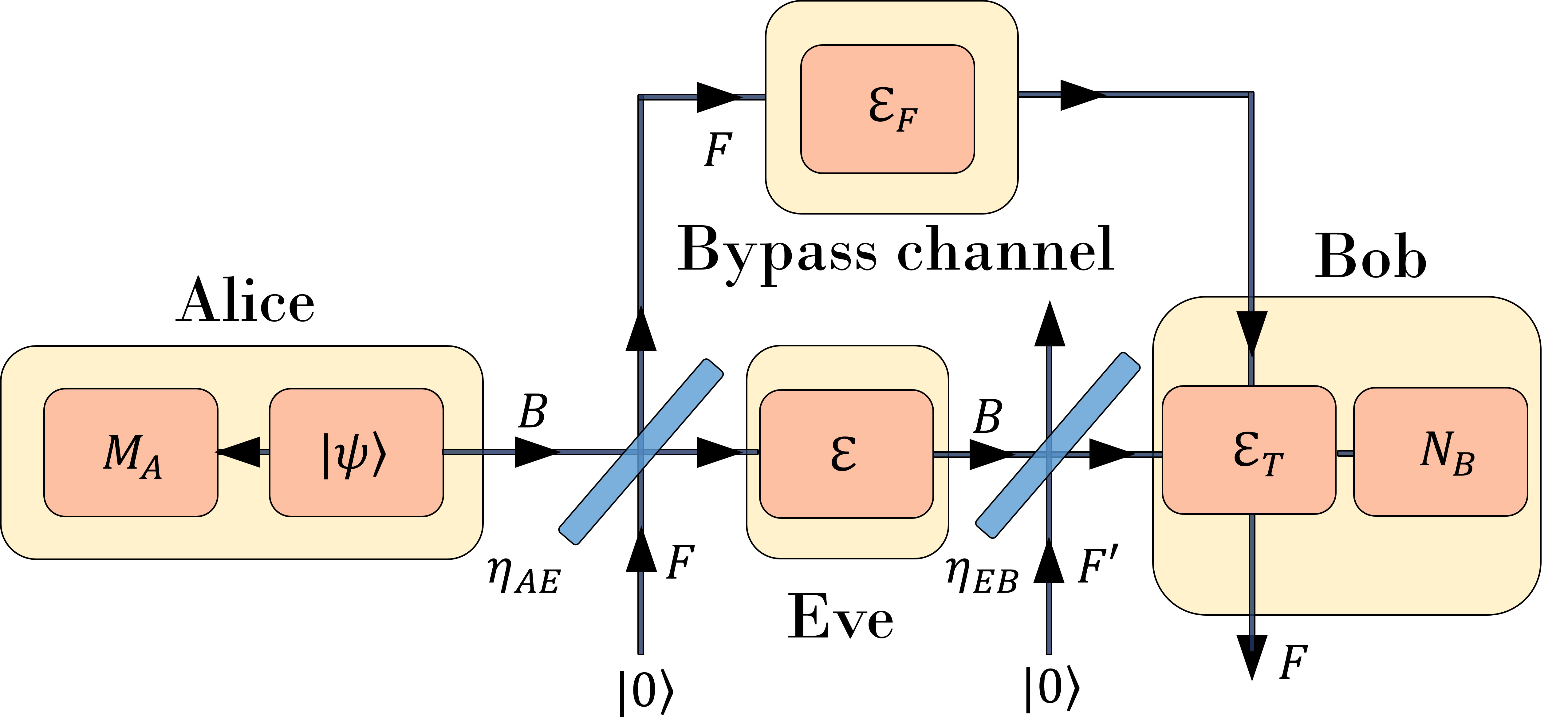}
\caption{The general bypass channel model. Each round, Alice prepares an entangled state $\ket{\psi}$ and measures one part, whilst sending the other to Bob along spatial mode $B$. First, the signal is interfered with the vacuum at a beam splitter with transmissivity $\eta_{AE}$, which models the lossy channel between Alice and Eve, who is described by an unknown quantum channel $\mathcal{E}$. After Eve's attack, the signal is interfered with another mode $F'$ at a beam splitter with transmittance $\eta_{EB}$ modeling a lossy channel between Eve and Bob, before reaching Bob's collection device, modeled by a channel $\mathcal{E}_{T}$. Any signal reflected at the first beam splitter travels along the bypass mode $F$, and undergoes an unknown quantum channel $\mathcal{E}_{F}$, before} reaching Bob's collection device. Importantly, no party has access to modes $F$ or $F'$ throughout the protocol, and the transmissivities $\eta_{AE}$ and $\eta_{EB}$ can be bounded using monitoring techniques. $\mathcal{E}_{F}$ and $\mathcal{E}_{T}$ on the other hand must be optimized over to obtain a lower bound on the key rate.
\label{fig:bypass_gen}
\end{figure}

Security analysis in the presence of bypass channels was first conducted in Ref.~\cite{ghalaii2023satellitebased}, and key rate improvements were found in both continuous variable (CV) and discrete variable (DV) cases within certain regimes of operation. Moreover, an easy to compute upper bound was derived, which served as a good approximation for the CV case with reverse reconciliation. For the DV case, it was found that single photon (SP) BB84, typically regarded as the ideal QKD protocol, was unaffected by such restricted eavesdropping. This conclusion was supported by both an analytical and numerical lower bound. On the other hand, using weak coherent pulses (WCP), Ref.~\cite{ghalaii2023satellitebased} found improvements beyond the single photon rate for very small $\eta_{AE}$. Broadly, this is due to Alice and Bob capitalizing on pulses with large photon numbers without the threat of photon number splitting attacks, highlighting how bypass channels require a change in intuition from the standard case.  

This left open a number of questions concerning DV QKD with bypass channels. Firstly, in Ref.~\cite{ghalaii2023satellitebased}, it was conjectured that the security analysis of SP-BB84 in the presence of bypass channels is tight, and the numerical framework supporting this was left to a future manuscript. Next, it would be interesting to understand, using this framework, if the bypass channel can benefit the single photon case at all; for example, device imperfections beyond those considered by Ref.~\cite{ghalaii2023satellitebased} could give rise to non-trivial bypass behavior. Finally, can numerical techniques also provide bounds in the WCP case? Importantly, does this also lead to tighter bounds on the key rate when the restriction on Eve is weaker, e.g., at higher values of $\eta_{AE}$? This is a well motivated question, since it is typically easier to guarantee weaker restrictions on Eve in practice, and a channel loss of 30 to 40 dB between Alice and Eve, as was required by Ref.~\cite{ghalaii2023satellitebased} for key rate improvements, may be unrealistically high. 

In this follow up work, we address these questions. Namely, we formulate a numerical framework for calculating key rates of DV-QKD protocols in the presence of bypass channels. This is achieved by modifying the technique of Ref.~\cite{Winick_2018}, which allows one to calculate reliable lower bounds on the asymptotic key rate using semidefinite programs (SDPs). We then reproduce the SP result of Ref.~\cite{ghalaii2023satellitebased}, confirming that standard single photon BB84 does not benefit from bypass channels. Next, we exploit the versatility of the technique to show that in the SP regime, a bypass model does increase the key rate under certain device imperfections. When we account for large mismatches in efficiencies between Bob's detectors, we improve key rates beyond that of conventional QKD with mismatched detectors~\cite{Winick_2018}, as well as extend the largest tolerable mismatch. Moreover, these improvements can even happen at moderate restrictions on Eve. Finally, we extend our numerical analysis to the WCP case. This involves modifying the dimension reduction technique of Refs.~\cite{Upadhyaya2021M,Upadhyaya2021,Upadhyaya2022} to account for bypass effects. We then show that there exist some scenarios where the WCP protocol benefits from moderate restrictions on Eve, which significantly differs from the regime studied in~\cite{ghalaii2023satellitebased}. This improvement is most prominent when restricting the maximum number of photons received by Bob. We also apply the dimension reduction techniques to lift this assumption, however, the resulting key rates are diminished, leaving improvements of this approach to future work.    

The paper is structured as follows. In \cref{sec:setting} we mathematically describe photonic DV-QKD protocols in the presence of bypass channels. In \cref{sec:security} we show how this formulation is compatible with existing numerical techniques for calculating the key rate, with some modifications. \cref{sec:SP} then describes how this is applied to SP-BB84, by first reinforcing the result of Ref.~\cite{ghalaii2023satellitebased}, and then extending the model to account for mismatched detectors at the receiver. Next, \cref{sec:WCP} extends the framework to WCP-BB84 and presents some results under a particular noise model. A discussion is presented in~\cref{sec:discussion} along with open questions. 

\ifarxiv\section{QKD with a bypass channel}\label{sec:setting}\else\noindent{\it QKD with a bypass channel.|}\fi 

Throughout this work, the Hilbert space associated to a quantum system $A$ is denoted by $\mathcal{H}_{A}$. The set of linear operators and density operators on $\mathcal{H}_{A}$ are denoted by $\mathcal{L}(\mathcal{H}_{A})$ and $\mathcal{D}(\mathcal{H}_{A})$, respectively. Classical systems are denoted by non-italic capital letters, e.g., $\mathsf{A}$. 

\begin{figure}[h]
\includegraphics[width=8.4cm]{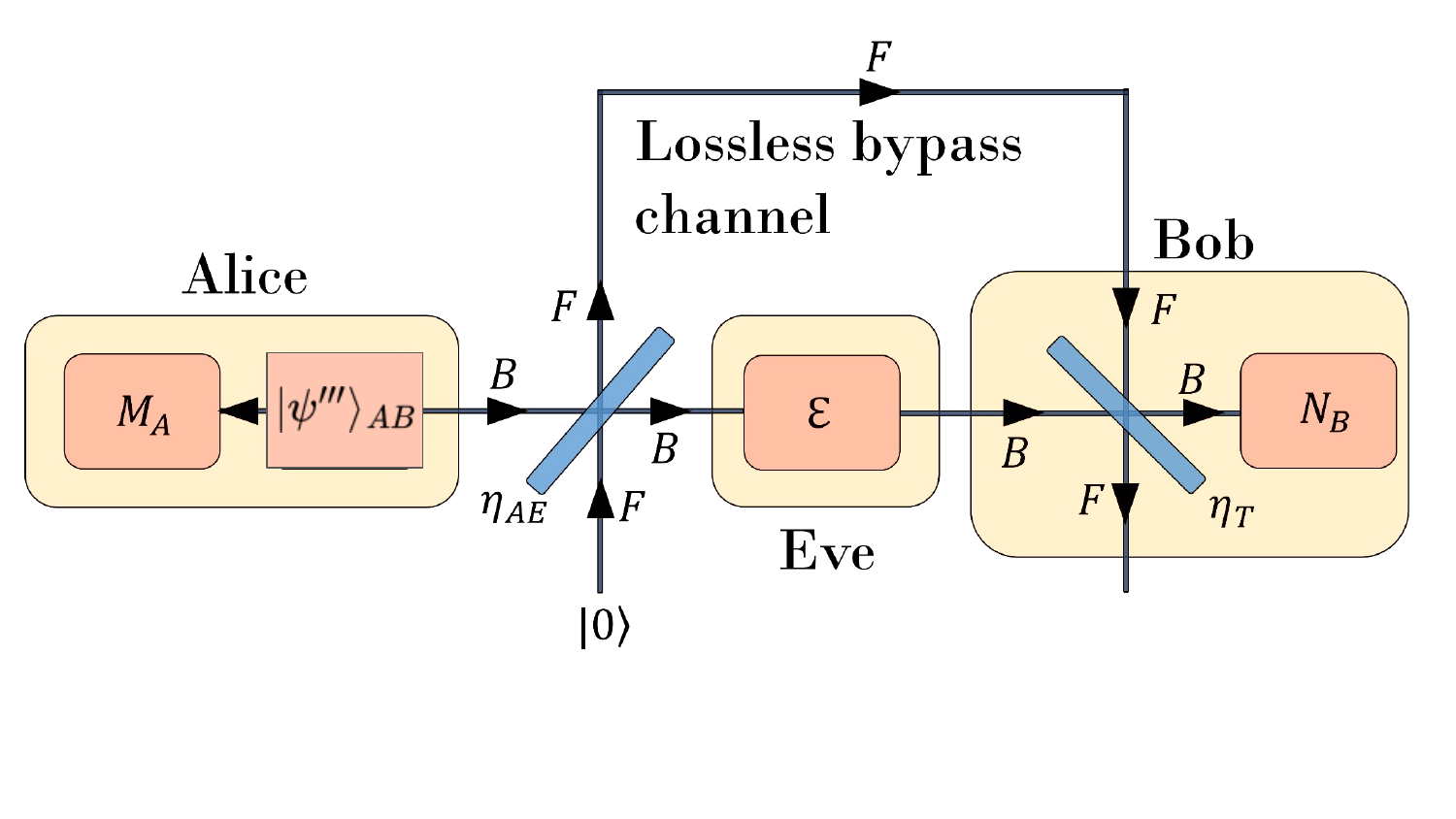}
\caption{The bypass channel model considered in this work. This model is a special case of \cref{fig:bypass_gen}, where the bypass channel is lossless (captured by choosing $\mathcal{E}_{F}$ as the identity channel), there is no loss between Eve and Bob (choosing $\eta_{EB} = 1$) and we model Bob's collection device $\mathcal{E}_{T}$ as a beam splitter with transmitivity $\eta_{T}$. Eve's attack is viewed as the action of a quantum channel $\mathcal{E}$ on mode $B$, which is related to the unitary description in the text via its Stinespring dilation $\mathcal{E}(\rho_{B}) = \tr_{E}[U_{BE}(\rho_{B} \otimes \ketbra{0}{0}_{E})U_{BE}^{\dagger}]$.}
\label{fig:bypass_lossless}
\end{figure}

Consider a photonic implementation of prepare and measure (P\&M) QKD in the presence of a lossless bypass channel, as shown in \cref{fig:bypass_lossless}. Note that this scenario is a special case of that shown in \cref{fig:bypass_gen}. However, it was observed in Ref.~\cite{ghalaii2023satellitebased} that a lossless bypass channel often favors Eve as it carries more of the burden that would have otherwise attributed to Eve. We therefore assume that the bypass channel is lossless in our work. We further take the case where $\eta_{EB} = 1$, and model Bob's collection device as a beam splitter with transmissivity $\eta_{T}$ (see \cite[Appendix C]{ghalaii2023satellitebased} for a justification of this model). These restrictions make the problem more manageable for our numerical approach, and we leave a more general treatment to future work.   

We model the global channel between Alice and Bob as a completely-positive-trace-preserving (CPTP) map acting on states in a photon number space $\mathcal{H}_{BF}$, with two spatial modes, $B$ and $F$, each with two polarization modes, $B_{H/V}$ and $F_{H/V}$. We denote by $E$ the system held by Eve, which is initialized in a state $\ket{0}_{E}$. Alice first prepares an entangled state $\ket{\psi'''}_{AB}$ between a reference system $A$, and spatial mode $B$. Specifically, $\mathcal{H}_{A} = \mathbb{C}^{2} \otimes \mathbb{C}^{2}$, corresponding to two basis choices, each with two outcomes (Alice's key bit)\footnote{Note that we have employed the source replacement scheme~\cite{BBM92,Ferenczi12}, where the system $A$ stores information about which state Alice prepared.}. We also assume mode $F$ is initially in the vacuum state, and denote the global initial state by
\begin{equation}
    \ket{\psi^{\text{init}}}_{ABFE} = \ket{\psi'''}_{AB} \otimes \ket{0}_{F} \otimes \ket{0}_{E} \in \mathcal{H}_{A} \otimes \mathcal{H}_{BF} \otimes \mathcal{H}_{E}. 
\end{equation}
Alice's measurement is described by the POVM $M_{A} = \{M_{a,x}\}_{a,x}$, where $M_{a,x} = \ketbra{a,x}{a,x}$ for $a,x \in \{0,1\}$, with $a$ representing the key bit and $x$ being the basis choice, and $\{\ket{i}\}_{i=0}^{d-1}$ denotes the standard basis for $\mathbb{C}^{d}$.

The two spatial modes $B$ and $F$ correspond to the two arms of the interferometer describing the channel between Alice and Bob. The polarization modes, $B_{H/V}$, are then used for encoding. Let $\eta_{AE}$ denote the transmitivity of the first beam splitter in the interferometer. This determines the signal fraction sent to Eve along arm $B$, and the fraction that bypasses her, along arm $F$. Eve then performs her action on $B$, which we describe by a unitary $U_{BE}$ on $\mathcal{H}_{BF} \otimes \mathcal{H}_{E}$, and both signals are recombined at a beam splitter with transmitivity $\eta_{T}$. Crucially, $U_{BE}$ acts trivially on mode $F$. The action of the channel is then described as follows:
\begin{equation}
    \begin{aligned}
    \ket{\psi^{\text{init}}}_{ABFE} &\xrightarrow[]{U_{1}} \ket{\psi''}_{ABF} \otimes \ket{0}_{E} \\ &\xrightarrow[]{U_{BE}} \ket{\psi'}_{ABFE} \\ &\xrightarrow[]{U_{2}} \ket{\psi}_{ABFE},
    \end{aligned} 
\end{equation}
where the unitaries $U_{1}$ and $U_{2}$ each describe a beam splitter on $BF$, with transmitivity $\eta_{AE}$ and $\eta_{T}$, respectively. Bob then measures mode $B$, modeled by a POVM on $\mathcal{H}_{BF}$, $N_{B} = \{N_{b,y}\}_{b,y}$, where $(b,y) \in \{0,1\}^{2} \cup \{(\perp,\perp)\}$ denotes the outcome of his measuring device. The outcome $(b,y) = (\perp,\perp)$ corresponds to a ``no-click'' event, whilst $(b,y) \in \{0,1\}^{2}$ corresponds to a basis choice $y$ and key bit $b$. ``Double-click'' events are randomly mapped to $b = 0$ or $b=1$ with equal probability, and hence not included as an explicit outcome of Bob's measurement device. Note that it is assumed no parties have access to the spatial mode $F$. Specifically, the POVM $N_{B}$ acts trivially on $F$. We call any tuple
\begin{equation}
    \big (\ket{\psi^{\text{init}}},U_{1},U_{BE},U_{2},M_{A},N_{B} \big)
\end{equation}
a quantum strategy for the bypass scenario. 

The post-measurement state shared between Alice, Bob and Eve is given by
\begin{equation}
    \rho_{\mathsf{ABXY}E} = \sum_{a,b,x,y} \ketbra{a,b}{a,b}_{\mathsf{AB}} \otimes \ketbra{x,y}{x,y}_{\mathsf{XY}} \otimes \rho_{E}^{a,b,x,y},
\end{equation}
where the sum ranges over possible outcomes of the measurements, $\mathsf{ABXY}$ are the classical registers storing those outcomes, and 
\begin{equation}
    \rho_{E}^{a,b,x,y} = \tr_{ABF}\big[ (M_{a,x} \otimes N_{b,y} \otimes \id_{E})\ketbra{\psi}{\psi}_{ABFE}\big]
\end{equation}
is Eve's post-measurement state conditioned on the outcome $(\mathsf{A},\mathsf{B},\mathsf{X},\mathsf{Y}) = (a,b,x,y)$. Alice and Bob then post-select on instances where Bob's device recorded a click, and their basis choice agrees, i.e., $\mathsf{X}= \mathsf{Y}$, which is described by the projector 
\begin{equation}
    \Pi = \sum_{x \in \{0,1\}} \id_{\mathsf{AB}} \otimes \ketbra{x,x}{x,x}_{\mathsf{XY}} \otimes \id_{E}.
\end{equation}
The re-normalized post-selected state is denoted by
\begin{multline}
    \tilde{\rho}_{\mathsf{ABXY}E}  = \frac{1}{\text{Pr}[\text{pass}]} \Pi \rho_{\mathsf{ABXY}E} \Pi \\
    = \frac{1}{\text{Pr}[\text{pass}]} \sum_{x \in \{0,1\}} \sum_{a,b \in \{0,1\}} \ketbra{a,b}{a,b}_{\mathsf{AB}} \\ \otimes \ketbra{x,x}{x,x}_{\mathsf{XY}} \otimes \rho_{E}^{a,b,x,x}, \label{eq:pms}
\end{multline}
where
\begin{equation}
    \text{Pr}[\text{pass}] = \tr[\Pi \rho_{\mathsf{ABXY}E} \Pi] = \sum_{x \in \{0,1\}} \sum_{a,b \in \{0,1\}} \tr[\rho_{E}^{a,b,x,x}]
\end{equation}
is the probability the post-selection filter is passed. 

\ifarxiv\section{Security analysis}\label{sec:security}\else\noindent{\it Security analysis.|}\fi 

\subsection{Key rate formula}

Given a post-selected state of a single round, $\tilde{\rho}_{\mathsf{ABXY}E}$, Alice and Bob can compute the secret key rate. In the asymptotic regime, the two relevant quantities are the conditional entropies $H(\mathsf{A}|\mathsf{XY}E)$ and $H(\mathsf{A}|\mathsf{BXY})$, where the latter quantifies the cost of Bob reconciling his key with Alice, and the former captures the fraction of Alice's key that is secret from Eve. The asymptotic key rate is then given by the Devetak-Winter formula~\cite{DW,Winick_2018}:
\begin{equation}
    r^{\infty} = \text{Pr}[\text{pass}]\Big( \inf_{\tilde{\rho} \in \tilde{\mathcal{F}}} \, H(\mathsf{A}|\mathsf{XY}E)_{\tilde{\rho}} \\ - H(\mathsf{A}|\mathsf{BXY}) \Big), \label{eq:rate_1}
\end{equation}
where both $\text{Pr}[\text{pass}]$ and $H(\mathsf{A}|\mathsf{BXY})$ can be calculated from the observed measurement statistics, whilst the infimum is taken over all post-measurement states $\tilde{\rho}$ compatible with those statistics, which we denote by a feasible set $\tilde{\mathcal{F}}$ to be defined later. Asymptotic rates can also be used as a basis for finite size security under general attacks using techniques such as the entropy accumulation theorem~\cite{DFR,ADFRV,Metger_2022,Metger_2023}. Thus, to prove security, we wish to lower bound $\inf_{\tilde{\rho} \in \tilde{\mathcal{F}}} \text{Pr}[\text{pass}] \, H(\mathsf{A}|\mathsf{XY}E)_{\tilde{\rho}}$. Towards this goal, we provide the following rewriting of the optimization problem.

\begin{lemma}
Let $\big (\ket{\psi^{\mathrm{init}}},U_{1},U_{BE},U_{2}, M_{A},N_{B} \big)$ be any quantum strategy in the bypass scenario that gives rise to the post-selected state $\tilde{\rho}$ in \cref{eq:pms}. Then the following holds:
\begin{multline*}
    \mathrm{Pr}[\mathrm{pass}] \, H(\mathsf{A}|\mathsf{XY}E)_{\tilde{\rho}} \\ = \sum_{x \in \{0,1\}} D\Big( V_{x}\rho_{ABF,x}'V_{x}^{\dagger} \big \| (\mathcal{Z}_{\tilde{A}} \otimes \mathcal{I}_{ABF}) \big[ V_{x}\rho_{ABF,x}'V_{x}^{\dagger} \big] \Big),
\end{multline*}
where
\begin{equation*}
\begin{aligned}
    \rho_{ABF,x}' &= \tr_{E}\big[\ketbra{\psi_{x}'}{\psi_{x}'}\big], \\ \ket{\psi_{x}'} &= \Bigg(\id_{A} \otimes U_{2}^{\dagger}\sqrt{\sum_{b \in \{0,1\}}N_{b,x}}U_{2} \otimes \id_{E}\Bigg)\ket{\psi'},
\end{aligned}
\end{equation*}
$V_{x}: \mathcal{H}_{A}\otimes \mathcal{H}_{BF} \to \mathcal{H}_{\tilde{A}}\otimes \mathcal{H}_{A}\otimes \mathcal{H}_{BF}$, where $\mathcal{H}_{\tilde{A}} \cong \mathbb{C}^{2}$, is defined by
\begin{equation*}
    V_{x} = \sum_{a \in \{0,1\}} \ket{a}_{\tilde{A}} \otimes M_{a,x} \otimes \id_{BF}
\end{equation*}
and satisfies $\sum_{x}V_{x}^{\dagger}V_{x} = \id_{ABF}$, and $\mathcal{Z}_{\tilde{A}}:\mathcal{L}(\mathcal{H}_{\tilde{A}}) \to \mathcal{L}(\mathcal{H}_{\mathsf{A}})$ is the pinching channel on $\tilde{A}$,
\begin{equation*}
    \mathcal{Z}_{\tilde{A}}[\sigma] = \sum_{a \in \{0,1\}} \bra{a}\sigma \ket{a}\ketbra{a}{a}_{\mathsf{A}}
\end{equation*}
for all $\sigma \in \mathcal{L}(\mathcal{H}_{\tilde{A}})$. \label{lem:relEnt}
\end{lemma}
\noindent See \cref{app:proof1} for proof. In \cref{lem:relEnt}, $D(\rho \| \sigma)$ is the quantum relative entropy between two states $\rho,\sigma \in \mathcal{D}(\mathcal{H})$, given by $\tr[\rho \log(\rho)] - \tr[\rho \log(\sigma)]$ when $\text{Supp}[\rho] \subseteq \text{Supp}[\sigma]$, and $+ \infty$ otherwise. 

Lemma \ref{lem:relEnt} allows us to simplify the optimization \eqref{eq:rate_1} in two ways. Firstly, the objective function is now independent of Eve's system, which a priori has an unknown dimension. Secondly, the states $\rho_{ABF,x}'$ occur before the final beam splitter $U_{2}$. This will allow us to constrain the problem.   

\subsection{Problem constraints}

During key exchange, Alice and Bob collect statistics of the form  
\begin{equation}
\begin{aligned}
    p_{a,b,x,y} &= \bra{\psi}(M_{a,x} \otimes N_{b,y} \otimes \id_{E})\ket{\psi} \\ &= \bra{\psi'}(M_{a,x} \otimes U_{2}^{\dagger}N_{b,y}U_{2} \otimes \id_{E})\ket{\psi'} \\
    &= \tr[(M_{a,x} \otimes N_{b,y}')\rho_{ABF}'],
\end{aligned}    \label{eq:const1}
\end{equation}
where $N_{b,y}' = U_{2}^{\dagger} N_{b,y} U_{2}$ are Bob's rotated POVM elements, and
\begin{equation}
    \rho_{ABF}' = \tr_{E}[\ketbra{\psi'}{\psi'}].
\end{equation}

We next account for the fact that Eve does not have access to system $A$ or mode $F$. Consider the following series of equalities:
\begin{equation}
\begin{aligned}
    \tr_{B}[\rho_{ABF}'] &= \tr_{BE}[\ketbra{\psi'}{\psi'}] \\
    &= \tr_{BE}[U_{BE}(\ketbra{\psi''}{\psi''}_{ABF} \otimes \ketbra{0}{0}_{E})U_{BE}^{\dagger}] \\
    &= \tr_{B}[\ketbra{\psi''}{\psi''}],
\end{aligned} \label{eq:ptrConst}
\end{equation}
where we used the cyclic property of the partial trace to remove the unitary $U_{BE}$, which only acts on systems $B$ and $E$\footnote{Note that when taking the partial trace over $B$ in \cref{eq:ptrConst}, we implicitly imposed a tensor product structure $\mathcal{H}_{BF} = \mathcal{H}_{B} \otimes \mathcal{H}_{F}$, where both $\mathcal{H}_{B}$ and $\mathcal{H}_{F}$ are Fock spaces.}. In other words, the reduced density operator for Alice's system and the bypass mode is invariant under Eve's attack. 

For $x \in \{0,1\}$, let 
\begin{equation}
    P_{x} := U_{2}^{\dagger}\sqrt{\sum_{b\in \{0,1\}}N_{b,x}}U_{2},
\end{equation}
and note that according to the definition in \cref{lem:relEnt},
\begin{equation}
    \rho_{ABF,x}' = P_{x} \rho_{ABF}' P_{x}. 
\end{equation}
We now have a sufficient number of constraints to characterize the space of feasible $\rho_{ABF}'$. By defining the completely-positive (CP) map $\mathcal{G}_{x}[\sigma_{ABF}] = V_{x}P_{x}\sigma P_{x} V_{x}^{\dagger}$, we arrive at the following optimization problem:
\begin{equation}
    \begin{gathered}
        \inf \ \sum_{x \in \{0,1\}} D\Big( \mathcal{G}_{x}[\rho_{ABF}'] \big \| (\mathcal{Z}_{\tilde{A}} \otimes \mathcal{I}_{ABF}) \big[ \mathcal{G}_{x}[\rho_{ABF}'] \big] \Big) \\
        \text{s.t.} \  \tr[(M_{a,x} \otimes N_{b,y}')\rho_{ABF}'] = p(a,b,x,y), \ \forall a,b,x,y, \\
         \tr_{B}[\rho_{ABF}'] = \tr_{B}[\ketbra{\psi''}{\psi''}], \\
         \rho_{ABF}' \in \mathcal{D}(\mathcal{H}_{A} \otimes \mathcal{H}_{B}\otimes \mathcal{H}_{F}). \label{eq:finalOpt}
    \end{gathered}
\end{equation}
We denote the feasible set of this optimization problem by $\mathcal{F}$ and label all the constraints (excluding normalization and positivity) by $\tr[\Gamma_{k}\rho_{ABF}'] = \gamma_{k}$ for an appropriate set of operators $\{\Gamma_{k}\}$ and real numbers $\{\gamma_{k}\}$. 

The optimization in \cref{eq:finalOpt} is analogous to that formulated in standard prepare and measure QKD~\cite{Winick_2018}, except for some key differences that account for bypass effects. Firstly, the problem is formulated in terms of the state before the final beam splitter, $\rho_{ABF}'$, rather than the state directly measured by Bob. The action of this beam splitter is accounted for by replacing Bob's POVM elements $N_{b,y}$ with the rotated ones $U_{2}^{\dagger}N_{b,y}U_{2}$. Secondly, we add source replacement constraints for both Alice's system and the bypass arm. This accounts for the action of the first beam splitter, and restricts the degrees of freedom accessible to Eve, i.e., she can only manipulate the spatial mode $B$.  

\subsection{Reliable lower bounds} \label{sec:lowbnds}
Developments in the security analysis of QKD protocols have employed numerical techniques to tackle optimization problems such as \cref{eq:finalOpt}~\cite{Coles_2016,Primaatmaja_2019,Winick_2018,Bunandar_2020,George_21,Zhou_22,tavakoli2023semidefinite,Ara_jo_2023,KS2026}. In particular, Ref.~\cite{Winick_2018} developed a robust technique to find reliable lower bounds using SDPs, which relies on two steps:
\begin{itemize}
    \item Step 1: An adapted Frank-Wolfe method~\cite{FW_56} is used to find a close to optimal upper bound. This corresponds to a near optimal attack by Eve, $\rho^{*}$.

    \item Step 2: The objective function is linearized at $\rho^{*}$. Due to the convexity of the objective function, minimizing the linearized objective function (which is an SDP) provides a lower bound on the original problem, and is tight when $\rho^{*}$ is optimal, i.e., Step 1 converged.
\end{itemize}
We can employ this technique to obtain numerical lower bounds on \cref{eq:finalOpt} when the Hilbert space $\mathcal{H}_{B} \otimes \mathcal{H}_{F}$ has a finite dimension. This follows from the fact that \cref{eq:finalOpt} is an instance of~\cite[Eq. (6)]{Winick_2018}. Specifically, the feasible set is the set of density operators that satisfy multiple trace constraints, and the objective function is of the form $D(\mathcal{G}[\rho]\|(\mathcal{Z}\circ \mathcal{G})[\rho])$ where $\mathcal{G}$ is a completely-positive (CP) map and $\mathcal{Z}$ is a unital CPTP map. When $\mathcal{H}_{B} \otimes \mathcal{H}_{F}$ is not finite dimensional, a situation commonly encountered in photonic realizations, we apply dimension reduction \cite{Upadhyaya2021,Upadhyaya2021M,Upadhyaya2022} to bound the infinite dimensional problem with a finite dimensional one (see \cref{sec:dr}). We then solve all SDPs using the Python package CVXPY~\cite{cvx1,cvx2}.   

We now elaborate further on how to apply the approach of Ref.~\cite{Winick_2018} to our problem. For Step 1, the Frank-Wolfe method described in \cite[Algorithm 1]{Winick_2018} can be employed, resulting in a near optimal attack $\rho^*$:
\begin{enumerate}
    \item Let $\delta > 0$, $\rho_{0} \in \mathcal{F}$ and set $i = 0$.
    \item Compute $\Delta \rho = \text{arg} \, \text{min}_{\Delta \rho} \tr[(\Delta \rho)^{\text{T}} \nabla f_{\epsilon}(\rho_{i})]$ subject to $\Delta \rho + \rho_{i} \in \mathcal{F}$. 
    \item If $|\tr[(\Delta \rho)^{\text{T}} \nabla f_{\epsilon}(\rho_{i})]| < \delta$ then output $\rho^* = \rho_{i}$. 
    \item If not, compute $\lambda = \text{arg} \, \text{min}_{\lambda} f_{\epsilon}(\rho_{i} + \Delta \rho)$ subject to $\lambda \in (0,1)$.
    \item Set $\rho_{i+1} = \rho_{i} + \lambda \Delta \rho$, $i \to i+1$ and go to 2.
\end{enumerate}
In the above, we defined the modified objective function $f_{\epsilon} : \mathcal{D}(\mathcal{H}_{A}\otimes \mathcal{H}_{BF}) \to \mathbb{R}$, $f_{\epsilon}(\rho) := \sum_{x\in \{0,1\}} f_{\epsilon}^{x}(\rho)$, where for $\epsilon > 0$,
\begin{equation}
    f_{\epsilon}^{x}(\rho) := D\Big( \mathcal{G}_{\epsilon}^{x}[\rho] \big \| (\mathcal{Z} \circ \mathcal{G}_{\epsilon}^{x})[\rho] \Big),
\end{equation}
$\mathcal{G}_{\epsilon}^{x} = \mathcal{D}_{\epsilon} \circ \mathcal{G}_{x}$, $\mathcal{Z} = \mathcal{Z}_{\tilde{A}}\otimes \mathcal{I}_{ABF}$ and $\mathcal{D}_{\epsilon}[\sigma] = (1-\epsilon)\sigma + \epsilon\id/\mathrm{dim}[\mathcal{H}]$ for $\sigma \in \mathcal{L}(\mathcal{H})$ is the depolarizing channel. By applying~\cite[Lemma 1]{Winick_2018}, the gradient of $f_{\epsilon}$ evaluated at a point $\rho$ exists for all $\rho \in \mathcal{D}(\mathcal{H}_{A}\otimes \mathcal{H}_{BF})$, and is given by $\nabla f_{\epsilon}(\rho)^{\text{T}} = \sum_{x} \nabla f_{\epsilon}^{x}(\rho)^{\text{T}} \in \mathcal{L}(\mathcal{H}_{A} \otimes \mathcal{H}_{BF})$, where
\begin{multline}
    \nabla f_{\epsilon}^{x}(\rho)^{\text{T}} = ((\mathcal{G}_{\epsilon}^{x})^{\dagger} \circ \log \circ \mathcal{G}_{\epsilon}^{x})[\rho] \\- ((\mathcal{G}_{\epsilon}^{x})^{\dagger} \circ \log \circ \mathcal{Z} \circ \mathcal{G}_{\epsilon}^{x})[\rho].
\end{multline}
Furthermore, the permuted function $f_{\epsilon}$ is continuous in $\epsilon$ for a fixed state \cite[Lemma 8]{Winick_2018}:
\begin{equation}
        |f_{0}(\rho) - f_{\epsilon}(\rho)| \leq \zeta_{\epsilon},
\end{equation}
where $\zeta_{\epsilon} = 2\epsilon(d'-1)\log \frac{d'}{\epsilon (d'-1)}$ and $d' = \mathrm{dim}[\mathcal{H}_{\tilde{A}ABF}]$. For our numerical examples, we choose $\epsilon = 10^{-8}$.

For Step 2, we employ \cite[Theorem 2]{Winick_2018} given a near optimal guess $\rho^*$ from Step 1. The result is a reliable lower bound on \cref{eq:finalOpt}, of the form
\begin{equation}
        f_{\epsilon}(\rho^*) - \tr[(\rho^*)^{\mathrm{T}} \nabla f_{\epsilon}(\rho^*)] + \beta - 2\zeta_{\epsilon}
    \end{equation}
    where $\beta$ is the solution to the following SDP:
    \begin{equation}
        \begin{aligned}
            \sup \ &\sum_{k} \gamma_{k} \, y_{k} \\
            \mathrm{s.t.} \ & \sum_{k} y_{k} \Gamma_{k}^{\mathrm{T}} \succeq \nabla f_{\epsilon}(\rho^*). 
        \end{aligned}
    \end{equation}

\subsection{Dimension Reduction}\label{sec:dr}

In photonic realizations of QKD, we rarely encounter sources that emit finite-dimensional states. Realistic sources (such as WCP) are instead modeled using a photon number space of unbounded dimension. Thus, to apply the previously discussed techniques to more realistic scenarios, we need to reduce an infinite dimensional problem to a finite one. The dimension reduction method \cite{Upadhyaya2021} achieves this by $(i)$ choosing a finite dimensional subspace, $(ii)$ bounding the weight of the state outside that subspace, $(iii)$ computing a correction term to the objective function and $(iv)$ constructing a relaxed finite dimensional feasible set. 

In the formulation of the problem in \cref{eq:finalOpt}, both $\mathcal{H}_B$ and $\mathcal{H}_F$ are infinite dimensional photon number spaces for two polarization modes. Hence, we need to choose a finite subspace of $\mathcal{H}_B \otimes \mathcal{H}_F$. Following the guidelines in Ref.~\cite{Upadhyaya2021}, choosing a projector that commutes with the POVMs defining the feasible set and with the POVMs defining the objective function gives a better lower bound on the key rate than a projector without these properties. 

First, we note that Bob's POVMs $N_{b,y}$ are block-diagonal in the total number of photons in $\mathcal{H}_{B} \otimes \mathcal{H}_{F}$ (see \cref{app:POVM} for details). That is, they satisfy 
\begin{align}
N_{b,y}&=\sum_{n=0}^\infty \sum_{k=0}^{n} (\Pi_{B}^{n-k} \otimes \Pi_{F}^{k})N_{b,y}(\Pi_{B}^{n-k} \otimes \Pi_{F}^{k}),
\end{align}
where
\begin{equation}
    \Pi_{B}^{k}:=\sum_{l=0}^k \ketbra{l_{B_{H}},(k-l)_{B_{V}}}
\end{equation}
and $\ket{k_{B_{H}},l_{B_{V}}} \in \mathcal{H}_{B_{H}} \otimes \mathcal{H}_{B_{V}} = \mathcal{H}_{B}$ is a Fock state. We define $\Pi_{F}^{k}$ similarly. Next, notice the final beam splitter is described by a unitary $U_{2}$, which preserves the total number of photons in the input modes, i.e., $U_2$ also satisfies 
\begin{align}
U_{2}&=\sum_{n=0}^\infty \sum_{k=0}^{n} (\Pi_{B}^{n-k} \otimes \Pi_{F}^{k})U_{2}(\Pi_{B}^{n-k} \otimes \Pi_{F}^{k}).
\end{align}
Thus, both $N_{b,y}$ and $U_2$ commute with any (linear combination of) projectors of the form $\Pi_{B}^{n-k} \otimes \Pi_{F}^{k}$. Consequently, we find the rotated measurements $N_{b,y}' = U_{2}^{\dagger}N_{b,y}U_{2}$ have the same commutation properties, as does the extension to Alice POVMs, $M_{a,x} \otimes N_{b,y}'$. These POVMs define both the key map and the constraints (cf.\ \cref{lem:relEnt}). 

Let $\Omega$  be any (finite) linear combination of the form $\Omega_{n} = \sum_{k=0}^{n}\Pi_{B}^{n-k} \otimes \Pi_{F}^{k}$. The above facts imply that, when choosing the projection $\Omega$ for the dimension reduction technique, the correction term is identically zero by \cite[Theorem 3]{Upadhyaya2021}, and we can use the loosened constraints of \cite[Theorem 4]{Upadhyaya2021}. Referring to \cite[Eq. (4.101)]{Upadhyaya2021M}, the following finite-dimensional optimization gives lower bounds on \cref{eq:finalOpt}:
\begin{equation}
    \begin{gathered}
        \inf \ \sum_{x \in \{0,1\}} D\Big( \mathcal{G}_{x}[\rho_{ABF}'] \big \| (\mathcal{Z}_{\tilde{A}} \otimes \mathcal{I}_{ABF}) \big[ \mathcal{G}_{x}[\rho_{ABF}'] \big] \Big) \\
        \text{s.t.} \ 1-W \leq \tr[\rho_{ABF}'] \leq 1\\
        p(a,b,x,y) -W \leq \tr[(M_{a,x} \otimes N_{b,y}')\rho_{ABF}'] \\ \hspace{5cm}\leq p(a,b,x,y), \forall a,b,x,y, \\
         \frac{1}{2}\big \| \tr_{B}[\rho_{ABF}'] - \tr_{B}[\ketbra{\psi''}{\psi''}] \big \|_1 \leq \sqrt{W}, \\
         \rho_{ABF}' \in \mathcal{P}(\mathcal{H}_{A} \otimes \Omega[\mathcal{H}_{BF}]). \label{eq:finalOptDR}
    \end{gathered}
\end{equation}
In the above, $\mathcal{P}(\mathcal{H})$ is the set of positive operators on $\mathcal{H}$, $\Omega[\mathcal{H}_{BF}]$ is the subspace of $\mathcal{H}_{B} \otimes \mathcal{H}_{F}$ obtained by applying the projector $\Omega$, and $W \in [0,1]$ bounds the weight of $\rho'_{ABF}$ outside the projected subspace. In particular, when $\Omega = \id$ we recover \cref{eq:finalOpt}. 

The weight $W$ should be estimated using the data collected by Alice and Bob during key exchange, such as the frequency of double click events by Bob's detectors~\cite{Zhang_2021,Li_2020,kamin2024}. This information can be used to upper bound the value of $W$ compatible with the observed statistics. In \cref{app:weight}, we demonstrate this for the special case $\eta_{T} = 1$ by deriving an analytical bound on $W$ in terms of the double click probability on Bob's side and the marginal state on $F$. We leave the question of tighter and more general bounds (i.e., ones that hold for $\eta_{T} < 1$) to future work. 

Furthermore, for a given choice of projection $\Omega$, defining the marginal $\tr_{B}[\tilde{\rho}_{ABF}']$ in \cref{eq:finalOptDR} is in fact non-trivial. This is a technical obstacle that is specific to the bypass scenario, and we elaborate on why this is the case and provide a resolution in \cref{app:pTrace}. We can then solve the dual formulation of \eqref{eq:finalOptDR} as described in \cite[Section V.B]{Upadhyaya2021}.

\ifarxiv\section{Application to BB84 with single photons}\label{sec:SP}\else\noindent{\it Application to BB84 with single photons.|}\fi

In this section, we use the tools developed in \cref{sec:security} to compute the key rate of the BB84 protocol with single photons in the presence of a bypass channel. We consider the protocol with active basis choice and polarization encoding. 

\subsection{Single photon bypass model} \label{sec:spbp}
For the single photon implementation, we consider the projector
\begin{equation}
    \Omega^{\text{sp}} = \Omega_{0} + \Omega_{1} = \Pi_{B}^{0} \otimes \Pi_{F}^{0} + \Pi_{B}^{0} \otimes \Pi_{F}^{1} + \Pi_{B}^{1} \otimes \Pi_{F}^{0}.
\end{equation}
The resulting Fock space on $BF$, $\mathcal{H}_{BF}^{\mathrm{sp}} = \Omega^{\text{sp}}[\mathcal{H}_{BF}]$, is given by $\mathrm{Sp} [\mathcal{B}_{1}]$, where
\begin{equation}
    \mathcal{B}_{1} = \{ \ket{0}, \ket{1_{B_{H}}} , \ket{1_{B_{V}}}, \ket{1_{F_{H}}} , \ket{1_{F_{V}}} \} \label{eq:SPH}
\end{equation}
is the $\leq 1$ photon Fock basis for two spatial modes $B,F$, each with two polarization modes $H/V$. For simplicity, we assume throughout this section that the weight $W = 0$. Note that this cannot be guaranteed in practice, and it corresponds to assuming that all parties can send and receive at most one photon. This may set additional restriction on Eve. Consequently, the optimization problem that we solve here could offer optimistic results as compared to the cases when $W>0$.  We will consider relaxing this assumption in \cref{sec:WCP}. 

We consider an initial state given by $\ket{\psi'''} = \ket{\psi_{\text{sp}}}$ where
\begin{equation}
    \ket{\psi_{\mathrm{sp}}}_{ABF} := \sum_{a,x \in \{0,1\}} \sqrt{p_{x}/2} \ket{a,x}_{A} \otimes \ket{1_{q(x,a)}}_{BF}, \label{eq:spstate}
\end{equation}
where $p_{0} = p_{z}$ and $p_{1} = 1-p_{z}$, $p_{z}\in (0,1)$, are the probabilities of choosing HV and DA basis, respectively,
\begin{equation}
    \big(q(0,0),q(0,1),q(1,0),q(1,1)\big) = \big(B_{H},B_{V},B_{+},B_{-}\big),
\end{equation}
and $\ket{1_{B_{\pm}}} = (\ket{1_{B_{H}}} \pm \ket{1_{B_{V}}})/\sqrt{2}$ are the DA basis states. The first beam splitter of the interferometer is described by the following unitary on $BF$ 
\begin{equation}
    U_{1} = \begin{bmatrix}
        1 & 0 & 0 & 0 & 0\\
        0 & \sqrt{\eta_{AE}} & 0 & \sqrt{1-\eta_{AE}} & 0\\
        0 & 0 & \sqrt{\eta_{AE}} & 0 & \sqrt{1-\eta_{AE}}\\
        0 & \sqrt{1-\eta_{AE}} & 0 & -\sqrt{\eta_{AE}} & 0\\
        0 & 0 & \sqrt{1-\eta_{AE}} & 0 & -\sqrt{\eta_{AE}}\\
\end{bmatrix},
\end{equation}
in the basis $\mathcal{B}_{1}$. We then have $\ket{\psi''} = \ket{\psi''_{\mathrm{sp}}} := (\id_{A} \otimes U_{1})\ket{\psi_{\mathrm{sp}}}$, which specifies the source replacement constraints in \cref{eq:finalOpt}. 

Eve's action is described by a quantum channel from the single photon subspace $\mathcal{H}_{BF}^{\mathrm{sp}}$ to itself. Note that this implies a squashing assumption. Namely, Eve can only receive and send at most one photon. The action of the second beam splitter is then given by
\begin{equation}
    U_{2} = \begin{bmatrix}
        1 & 0 & 0 & 0 & 0\\
        0 & -\sqrt{\eta_{T}} & 0 & \sqrt{1-\eta_{T}} & 0\\
        0 & 0 & -\sqrt{\eta_{T}} & 0 & \sqrt{1-\eta_{T}}\\
        0 & \sqrt{1-\eta_{T}} & 0 & \sqrt{\eta_{T}} & 0\\
        0 & 0 & \sqrt{1-\eta_{T}} & 0 & \sqrt{\eta_{T}}\\
\end{bmatrix}, \label{eq:bs2}
\end{equation}
in the basis $\mathcal{B}_{1}$. Bob measures a five outcome POVM given by the elements
\begin{equation}
    \begin{gathered}
        N_{0,0} = p_{0}\ketbra{1_{B_{H}}}{1_{B_{H}}}, \
        N_{1,0} = p_{0}\ketbra{1_{B_{V}}}{1_{B_{V}}}, \\
        N_{0,1} = p_{1}\ketbra{1_{B_{+}}}{1_{B_{+}}}, \
        N_{1,1} = p_{1}\ketbra{1_{B_{-}}}{1_{B_{-}}}, \\
        N_{\perp,\perp} = \id_{BF}^{\text{sp}} - \sum_{b,y \in \{0,1\}}N_{b,y},
    \end{gathered} \label{eq:spPOVM}
\end{equation}
allowing us to define the rotated measurements $N_{b,y}' = U_{2}^{\dagger} N_{b,y} U_{2}$. Note that the above measurement does not require photon number resolving (PNR) detectors. Rather, the single photon assumption implies Bob can only detect at most one photon, hence the measurement description of threshold and PNR detectors coincide. For simplicity, we have not included dark counts in the POVM definition above.

To constrain the problem, we consider the quantum bit error rate (QBER), which we assume is the same in each basis, denoted by $\text{E}$, and the total single click probability, $\text{Q}$. The corresponding operators on $\mathcal{H}_{A} \otimes \mathcal{H}_{BF}^{\text{sp}}$ are given by
\begin{equation}
    \begin{aligned}
        E_{Z}' &= \ketbra{0,0}{0,0} \otimes N_{1,0}' + \ketbra{1,0}{1,0} \otimes N_{0,0}', \\
        E_{X}' &= \ketbra{0,1}{0,1} \otimes N_{1,1}' + \ketbra{1,1}{1,1} \otimes N_{0,1}', \\
        E_{\varnothing}' &= \id_{A} \otimes N_{\perp,\perp}'.
    \end{aligned}    \label{eq:rotMeas}
\end{equation}
We then consider the following constraints on any feasible $\rho \in \mathcal{D}[\mathcal{H}_{A} \otimes \mathcal{H}_{BF}^{\mathrm{sp}}]$:
\begin{equation}
    \begin{aligned}
        \tr[E_{Z}'\rho] &= p_{0}^{2}\text{Q} \cdot \text{E}, \\
        \tr[E_{X}'\rho] &= p_{1}^{2}\text{Q} \cdot \text{E}, \\
         \tr[E_{\varnothing}'\rho] &= 1-\text{Q}.
    \end{aligned}
\end{equation}
The error correction term is given by
\begin{equation}
    \text{Pr}[\text{pass}] H(\mathsf{A}|\mathsf{BXY}) = (1-2p_{z}(1-p_{z}))\mathrm{Q}H_{\text{bin}}(\text{E}),  \label{eq:EC1}
\end{equation}
where we have assumed for simplicity perfect error correction efficiency, $H_{\mathrm{bin}}$ is the Shannon binary entropy function and the pre-factor $(1-2p_{z}(1-p_{z}))$ is the sifting efficiency. The values of $\text{Q}$ and $\text{E}$ are determined by the experimental setup, and for simulation purposes we fix them according to a noise model. Specifically, we calculate values according to the formulas from~\cite[Appendix A]{Panayi_2014}:
\begin{equation}
\begin{aligned}
    \text{Q} &= 1 - (1-\eta_{\mathrm{ch}}\eta_{\mathrm{d}})(1-p_{d})^{2}, \ \  \text{and} \\
    \text{E} &= (e_{0}\text{Q} - (e_{0} - e_{d})\eta_{\mathrm{ch}}\eta_{\mathrm{d}} p_{d})/\text{Q},
\end{aligned} 
\end{equation}
where $e_{0} = 1/2$. The numerical values used are detailed in~\cref{tab:num}.

\begin{table}[h!]
\centering
\begin{tabular}{|c | c | c|} 
 \hline
 Symbol & Description & Value \\ [0.5ex] 
 \hline\hline
 $\eta_{\mathrm{ch}}$ & Channel transmissivity & 0.001 \\ \hline
 $\eta_{\mathrm{d}}$ & Detector efficiency & 0.9 \\\hline
 $e_{d}$ & \makecell{Misalignment probability} & 0.01 \\\hline
 $p_{d}$ & Dark count probability & \makecell{$10^{-7}$ \\ per pulse} \\\hline 
 $p_{z}$ & \makecell{Probability of choosing \\ $Z$ basis measurement} & 0.5\\
 [1ex] 
 \hline
\end{tabular} 
\caption{Simulation parameters for the single photon BB84 example with matched detectors.}
\label{tab:num}
\end{table}

For comparison, we consider single photon BB84 without a bypass channel. Here, a lower bound on the rate is given by~\cite{PIRANDOLA_2008}:
\begin{equation}
    r^{\infty} \geq (1-2p_{z}(1-p_{z}))\text{Q}\big[ 1 - 2H_{\mathrm{bin}}(\text{E}) \big]. \label{eq:SPnorm}
\end{equation}
Additionally, we consider the bound from Ref.~\cite{ghalaii2023satellitebased} that accounts for the bypass model, 
\begin{multline}
    r^{\infty} \geq (1-2p_{z}(1-p_{z}))\text{Q}\Big[ -H_{\mathrm{bin}}(\text{E}) \\ + \frac{S_{11}^{\mathrm{L}}}{\text{Q}}\big(1-H_{\mathrm{bin}}(\varepsilon_{11}^{\mathrm{U}})\big) + \frac{S_{0}^{\mathrm{L}}}{\text{Q}}\Big], \label{eq:extBP-SP}
\end{multline}
where $S_{0}^{\mathrm{L}} = \max\{\text{Q}-\eta_{AE},0\}$, $S_{11}^{\mathrm{L}} = \max\{\text{Q}-(1-\eta_{AE}),0\}$ and $\varepsilon_{11}^{\mathrm{U}} = \min\{\text{E} \cdot \text{Q}/S_{11}^{\mathrm{L}},1/2\}$. We refer the reader to Ref.~\cite{ghalaii2023satellitebased} for the physical meanings of these quantities. Moreover, as discussed Ref.~\cite{ghalaii2023satellitebased}, the bound in \cref{eq:SPnorm} also serves as a valid bound in the presence of bypass channels, and we take the maximum of the two for our comparison with the numerics. 

Importantly, the observed statistics are independent of the bypass parameters $\eta_{AE}$ and $\eta_{T}$. Rather, we fix the observations $\text{Q}$ and $\text{E}$, and consider the set of $(\eta_{AE},\eta_{T})$ for which the optimization in \cref{eq:finalOpt} is feasible. This corresponds to the set of bypass models compatible with observations. Motivated by the fact that Alice and Bob may, through some means, bound $\eta_{AE}$~\cite{ghalaii2023satellitebased}, but have no information about $\eta_{T}$, our analysis will focus on fixing $\eta_{AE}$ and minimizing the rate over feasible $\eta_{T}$.  

In \cref{fig:SP1} we plot our numerical bound against the existing lower bound for bypass channels, and the bound corresponding to single photon BB84 without a bypass model. Similar to Ref.~\cite{ghalaii2023satellitebased}, here, these two coincide. At a given $\eta_{AE}$, we found that feasible values of $\eta_{T}$ lie roughly in the interval $[1-\eta_{AE},1]$, and observe a rate consistent with \cref{eq:SPnorm} over the entire interval. We therefore set $\eta_{T} \approx 1$ for the figure. 

As aforementioned, Ref.~\cite{ghalaii2023satellitebased} derived two analytical bounds for the single photon case. The first is tailored to the bypass situation, and the other is given by a standard formula. The latter was shown to perform better, and whether this was optimal was left open. Our results in \cref{fig:SP1} suggests this to be the case, as we could find no advantage from considering bypass channels in standard single photon BB84. 

Furthermore, the analytical bound from Ref.~\cite{ghalaii2023satellitebased} does not make a single photon assumption on the system sent from Eve to Bob. As described at the beginning of this subsection, our analysis does make this assumption; this was a consequence of choosing the weight $W = 0$. While because of the limitations of our numerical approach, we cannot rule out the possibility of finding a lower key rate when $W>0$, in the case $W=0$, we observe that the explicit attack $\rho^*$ obtained in Step 1 is nearly optimal\footnote{By optimal, we mean that the key rate at Step 1 matches that of Step 2.}, and no significant improvements are found by running a more intensive numerical calculation. Consequently, we expect that our numerical bounds are almost optimal, and the fact that they closely agree with the lower bound of Ref.~\cite{ghalaii2023satellitebased} provides evidence of tightness. 

\begin{figure}[h]
\includegraphics[width=8.4cm]{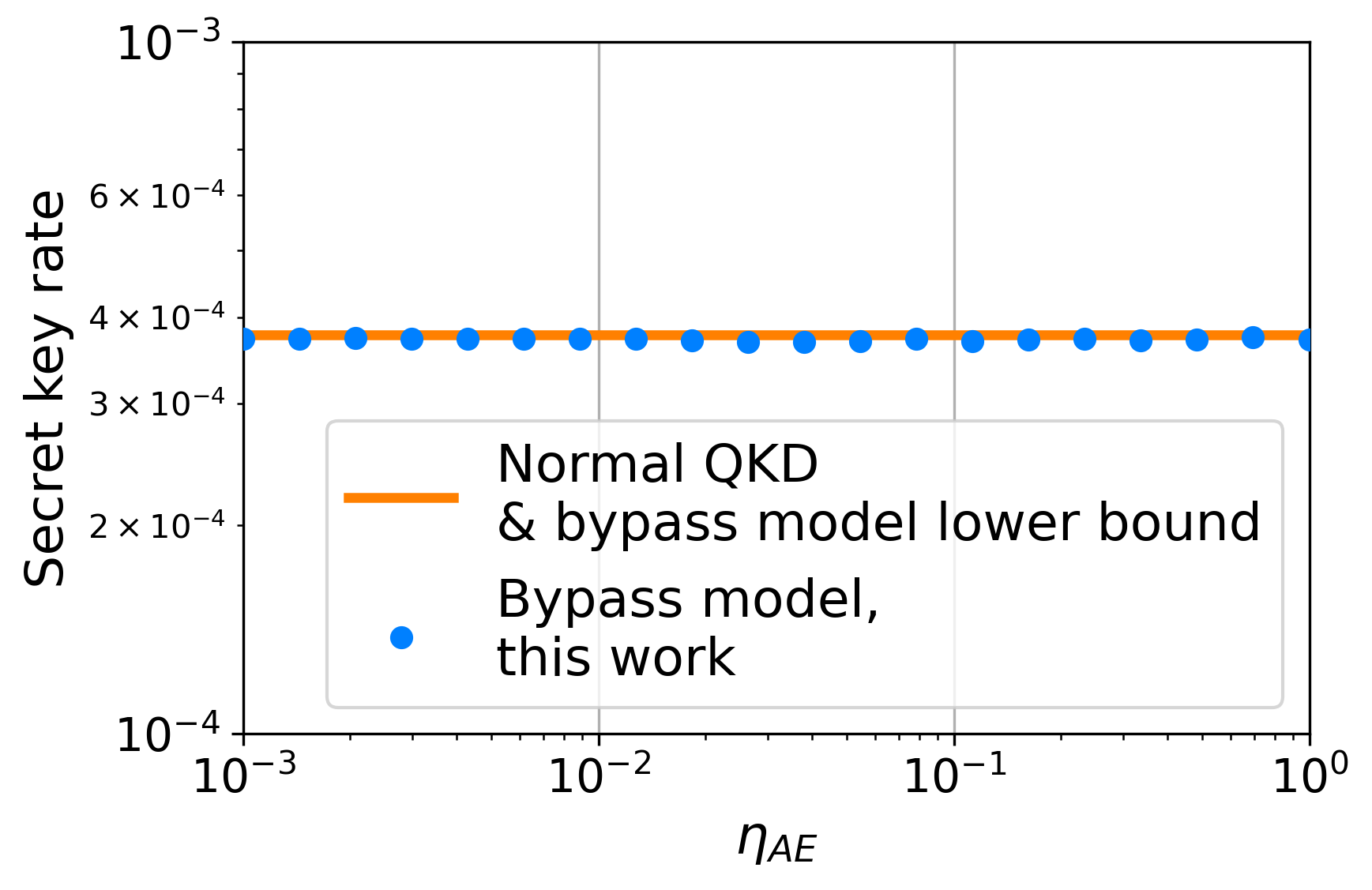}
\caption{Secret key rate of single photon BB84 in the presence of bypass channels. $\eta_{AE}$ denotes the fraction of signal sent from Alice and received by Eve, which can be bounded via monitoring techniques. ``Normal QKD'' indicates the rate in the absence of bypass channels, and ``bypass model lower bound'' refers to that derived in~\cite{ghalaii2023satellitebased}; here both coincide. ``Bypass model this work'' denotes the results of our numerical calculation; we observe an invariant rate as we vary the bypass parameter $\eta_{T}$, which we hence set to $1$ in this plot. Simulation parameters are given in \cref{tab:num}}
\label{fig:SP1}
\end{figure}

\subsection{Single photon with mismatched detector efficiencies}
One straightforward adaptation to the above model is the case where Bob's detectors have unmatched efficiencies $\eta_{1}$ and $\eta_{2}$. Such a scenario arises practically~\cite{Fung2009,Trushechkin2022, Zhang_2021,Winick_2018}, and the single photon case has been studied both numerically and analytically~\cite{Winick_2018,Bochkov_19}. Here, we extend this analysis to the bypass scenario.  

To incorporate mismatched detector efficiencies, we need to modify Bob's POVM elements:
\begin{equation}
    \begin{aligned}
        N_{0,0} \mapsto \eta_{1} N_{0,0}, \ N_{1,0}\mapsto \eta_{2} N_{1,0}, \\
        N_{0,1} \mapsto \eta_{1} N_{0,1}, \ N_{1,1}\mapsto \eta_{2} N_{1,1},
    \end{aligned}
\end{equation}
with $N_{\perp,\perp} = \id_{BF}^{\text{sp}} - \sum_{b,y \in \{0,1\}}N_{b,y}$ and $N_{b,y}' = U^{\dagger}_{2}N_{b,y}U_{2}$ as before. 

We also align our noise model with that detailed in~\cite[Appendix F.1]{Winick_2018}. Let $\mathcal{K}:\mathcal{D}[\mathcal{H}_{BF}^{\mathrm{sp}}] \to \mathcal{D}[\mathcal{H}_{BF}^{\mathrm{sp}}]$ define the depolarizing channel, 
\begin{equation}
    \mathcal{K}(\sigma) = (1-q)\sigma + q(\ketbra{1_{B_{H}}}{1_{B_{H}}} + \ketbra{1_{B_{V}}}{1_{B_{V}}})/2 \label{eq:depol}
\end{equation}
for $q \in [0,1]$. We then consider the noisy state
\begin{equation}
    \rho^{\mathrm{sim}}_{\mathrm{sp}} = (\mathcal{I}_{A} \otimes \mathcal{K})[\ketbra{\psi_{\mathrm{sp}}}{\psi_{\mathrm{sp}}}],
\end{equation}
and compute the statistics 
\begin{equation}
\begin{aligned}
    \text{E}_{Z} &:= \tr[E_{Z}\rho^{\mathrm{sim}}_{\mathrm{sp}}] \\
    \text{E}_{X} &:= \tr[E_{X}\rho^{\mathrm{sim}}_{\mathrm{sp}}], \\
    \text{Q} &:= \tr[E_{\varnothing}\rho^{\mathrm{sim}}_{\mathrm{sp}}],
\end{aligned}
\end{equation}
where $E_{Z},E_{X},E_{\varnothing}$ are defined as in \cref{eq:rotMeas} with $N_{b,y}'$ replaced with the un-rotated measurements $N_{b,y}$. Note that the error correction term is given analogously to \cref{eq:EC1}, by replacing $\text{E}$ with $E_Z$. The constraints on any feasible $\rho \in \mathcal{D}[\mathcal{H}_{ABF}^{\mathrm{sp}}]$ are then given by $\tr[E_{w}'\rho] = \text{E}_{w}$ for $w \in \{Z,X\}$ and $\tr[E_{\varnothing}'\rho] = \text{Q}$.

Note, as before, we fix these statistics independently of $\eta_{AE}$ and $\eta_{T}$ so they only depend on the noise model, specified by $q$, $\eta_{1}$ and $\eta_{2}$. For our numerics we choose the case of $q = 0.1$ and $\eta_{2} = 1$, and examine the bypass behaviour across the range $\eta_{1}\in(0,1]$, corresponding to an efficiency mismatch. 

\begin{remark}
    We will consider the case where, for a given $\eta_{AE}$, we do not minimize the key rate over all feasible $\eta_{T}$, and instead choose $\eta_{T} = \eta_{AE}$. We refer to these results as ``heuristic upper bounds'' on the rate at a given $\eta_{AE}$, which reflects the fact that they are not rigorous. This is because our numerical approach provides a lower bound on the key rate for a fixed pair $(\eta_{AE},\eta_{T})$, which is not guaranteed to be tight. However, we observe in practice that our numerical bounds are in fact almost tight\footnote{This can be inferred when the value of the objective function for the near optimal guess, $f(\rho^*)$, obtained in Step 1, roughly matches the solution to the linearized SDP in Step 2, as described in \cref{sec:lowbnds} and Ref.~\cite{Winick_2018}.}, and we can use these results to gain an insight into the possible improvements arising from the bypass model. \label{rem:ub}
\end{remark}
We divide the numerical results into two pairs of figures which are summarized below.
\begin{itemize}
    \item In \cref{fig:SPDet1}, we consider the case $\eta_{AE} = \eta_{T}$, corresponding to heuristic upper bounds on the key rate. We find that for various detector efficiency mismatches, a larger restriction on Eve (corresponding to a smaller value of $\eta_{AE}$) results in a larger key rate. This improvement becomes more pronounced as the mismatch increases. 
    \item In \cref{fig:SPDet2}, we minimize the rate over $\eta_{T}$ for each $\eta_{AE}$ and a fixed detector efficiency mismatch. We find that the key rate increases as the restriction on Eve increases. We thus conclude that improvements can be found from the bypass model when there is a large detector efficiency mismatch, contrasting the matched case explored in \cref{sec:SP}.   
\end{itemize}

\subsubsection{Heuristic upper bounds}

When $\eta_{T} = \eta_{AE}$, we find the existence of a bypass channel for a broad range of values for $\eta_{AE}$. The resulting key rate is plotted in \cref{fig:SPDet1}. Contrasting the case of matched efficiencies, the plot shows nontrivial bypass behaviour. At large mismatched efficiencies, we observe a vanishing key rate for normal QKD ($\eta_{AE}=\eta_{T}=1$), yet the heuristic upper bound when $\eta_{AE}<1$ is non-zero, and increasing in $1-\eta_{AE}$, suggesting room for improvement. We also find the possibility of reaching higher values of depolarizing noise at lower $\eta_{AE}$ in \cref{fig:SPDet1}(b).    
\begin{figure}[h]
   \centering \includegraphics[width=0.9\linewidth]{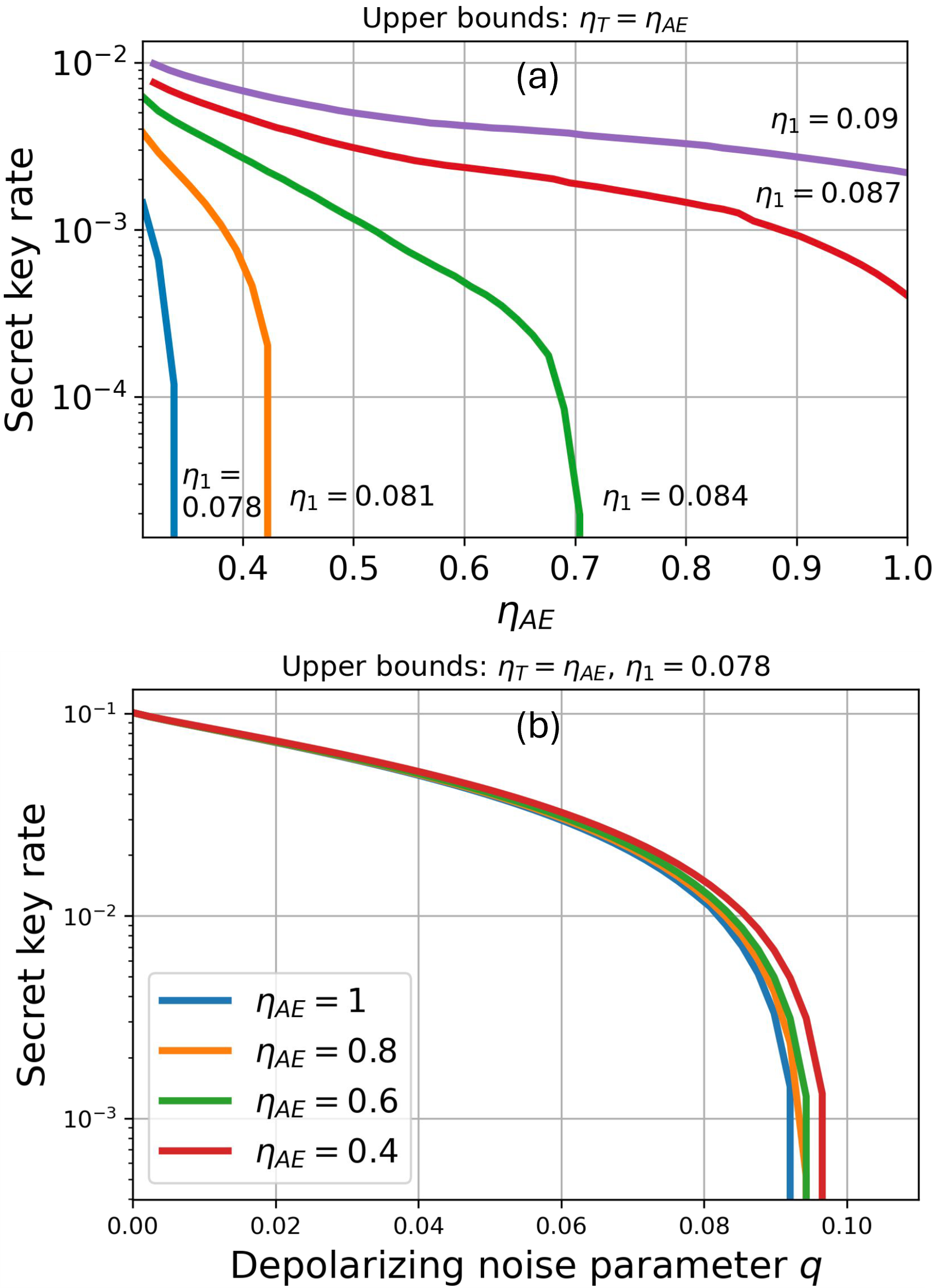}
\caption{Secret key rate of single photon BB84 with detector efficiency mismatch in the presence of bypass channels. $\eta_{1}$ is the efficiency of Bob's first detector, while the other has unity efficiency. We set the bypass parameter $\eta_{T}=\eta_{AE}$; hence the curves serve as heuristic upper bounds on the key rate. In (a), we set $q=0.1$ in \cref{eq:depol} to calculate the observed statistics, and vary $\eta_{AE}$ at different detector efficiencies. In (b), we take different values of $\eta_{AE}$, and vary the noise $q$ at a single detector efficiency.} 
\label{fig:SPDet1}
\end{figure}

In \cref{fig:SPDet1}(b), all the rates converge to the same rate as the depolarizing noise $q$ tends to 0. Moreover, we observe that when $q = 0$ and $\eta_{1} < 1$, any feasible pair $(\eta_{AE},\eta_{T})$ results in the same rate, similar to the single photon behaviour without a mismatch discussed in \cref{sec:spbp}. In the honest implementation with $q = 0$, the case of mismatched detectors is equivalent to preparing a tilted entangled state\footnote{By tilted, we refer to an entangled state that has, for example, uneven weights on the $\ket{00}$ and $\ket{11}$ components.} as opposed to one that is maximally entangled~\cite{Fung2009}. In this case, the rate is dependent only on the amount of efficiency mismatch (which influences the amount of tilting), owing to the fact that any passive adversary can always guess the outcome associated to the higher efficiency. The rate in this case is given by 
\begin{equation}
    r(\eta_{1},\eta_{2}) = \frac{\eta_{1} + \eta_{2}}{2}(1-2p_{z}(1-p_{z}))H_{\mathrm{bin}}\Big( \frac{\max\{\eta_{1},\eta_{2}\}}{\eta_{1} + \eta_{2}}\Big),
\end{equation}
where the factor of $\frac{\eta_{1} + \eta_{2}}{2}$ is the probability that either detector clicks given a single photon is sent down the channel carrying a random binary encoding, and the argument $\frac{\max\{\eta_{1},\eta_{2}\}}{\eta_{1} + \eta_{2}}$ quantifies the detector efficiency mismatch. Indeed, for the case $\eta_{1} = 0.078, \  \eta_{2} = 1$ and $p_{z} = 1/2$ we find $r(\eta_{1},\eta_{2}) \approx 0.1$, as seen numerically in \cref{fig:SPDet1}(b) when $q = 0$.

\subsubsection{Lower bounds}

To obtain lower bounds, we minimize the rate over feasible $\eta_{T}$ at a given noise level and given $\eta_{AE}$, corresponding to the worst case scenario. These results are plotted in \cref{fig:SPDet2}. At each $(\eta_{1},\eta_{AE})$, we find a convex curve as we plot the rate versus $\eta_{T}$ (see \cref{fig:SPDet2}(b)), and the minimum is higher than that of normal QKD. Notably, lower values of $\eta_{AE}$ allow for a nonzero key rate where the rate is zero for normal QKD. That is, bypass channels allow us to reach larger mismatches in efficiency than what is possible in normal QKD. 

\begin{figure}[h]
\centering
    \includegraphics[width=0.9\linewidth]{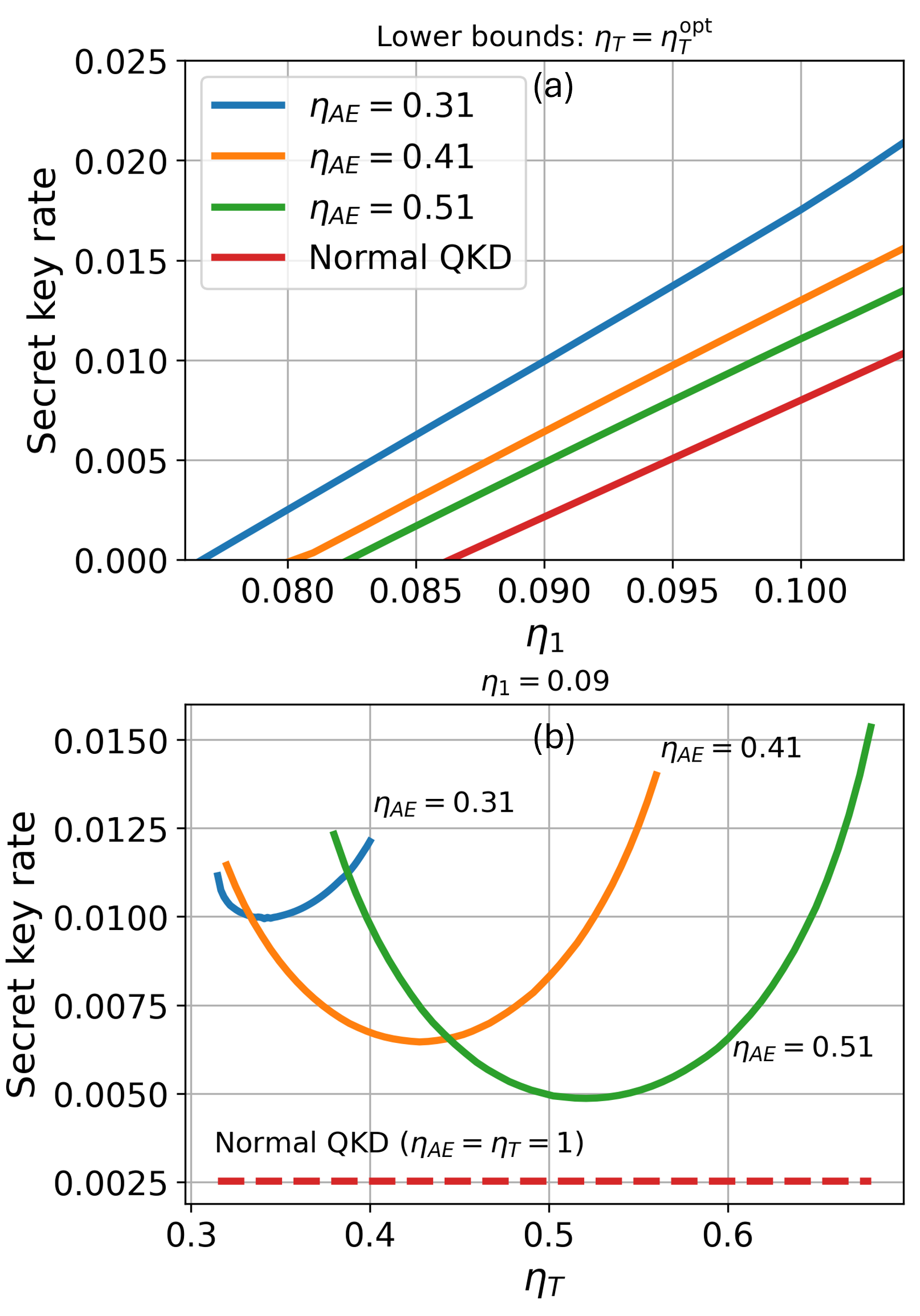}  
    \caption{Secret key rate of single photon BB84 with detector efficiency mismatch, in the presence of bypass channels. $\eta_{1}$ is the efficiency of Bob's first detector. The efficiency of Bob's second detector is fixed at $\eta_{2}=1$. In (a) for any fixed value of $\eta_{AE}$ and $eta_{1}$ the rate is minimized over feasible values of $\eta_{T}$. In (b), we show the nature of this minimization at $\eta_{1}=0.09$. The rates are compared to a normal QKD key rate derived from our numerical framework when $\eta_{AE} = \eta_{T} = 1$, highlighting improvements from the bypass channel. The noise parameter is set at $q=0.1$ for both plots.}
\label{fig:SPDet2}
\end{figure}  

\ifarxiv\section{Application to BB84 with weak coherent pulses}\label{sec:WCP}\else\noindent{\it Application to BB84 with weak coherent pulses.|}\fi 

Performing BB84 with single photons is challenging in practice, and a popular alternative is to use phase randomized weak coherent pulses (WCP). Adapting security proofs to account for this is therefore essential. In this section, we apply the tools developed in \cref{sec:security} to study WCP-BB84 in the presence of bypass channels. We present some initial findings, and motivate future research on the problem.
\subsection{WCP bypass model}
We begin by denoting the full four mode Fock space for systems $BF$ as $\mathcal{H}_{B} \otimes \mathcal{H}_{F} = \mathcal{H}_{BF}^{\mathrm{wcp}}= \mathrm{Span}[\mathcal{B_{\infty}}]$, where
\begin{equation}
    \mathcal{B}_{\infty} = \big\{ \ket{n_{B_{H}},m_{B_{V}},k_{F_{H}},l_{F_{V}}}\big\}_{n,m,k,l\in\{0,1,...\}} 
\end{equation}
is the basis for a Fock space with two spatial modes $B$ and $F$, each with two polarization modes $H$ and $V$. We label the field operators for the first and second spatial mode as $\hat{b}_{H/V}$ and $\hat{f}_{H/V}$ respectively, and define $\hat{b}_{\pm} = (\hat{b}_{H} \pm \hat{b}_{V})/\sqrt{2}$ and similarly $\hat{f}_{\pm}$. A coherent state with parameter $\alpha_{q(a,x)} \in \mathbb{C}$ is denoted by
\begin{equation}
    \ket{\alpha_{q(a,x)}} = \sum_{n=0}^{\infty}\frac{(\alpha_{q(a,x)})^{n}}{\sqrt{n!}}\ket{n_{q(a,x)}},
\end{equation}
and the initial entangled state prepared by Alice is given by
\begin{equation}
    \ket{\psi_{\mathrm{wcp}}}_{ABF} := \sum_{a,x \in \{0,1\}}\sqrt{p_{x}/2}\ket{a,x}_{A}\ket{\alpha_{q(a,x)}}_{BF}, \label{eq:WCPstate}
\end{equation}
where $p_{x}$ and $q(a,x)$ are defined below \cref{eq:spstate}. Next, the state interacts with the first beam splitter, $U_{1}$, which performs the mapping $\hat{b}_{q} \mapsto \sqrt{\eta_{AE}}\, \hat{b}_{q} + \sqrt{1-\eta_{AE}}\, \hat{f}_{q}$, resulting in $\ket{\psi'_{\mathrm{wcp}}} := (\id\otimes U_{1})\ket{\psi_{\mathrm{wcp}}}$ where
\begin{multline}
    (\id\otimes U_{1})\ket{\psi_{\mathrm{wcp}}} = \sum_{a,x \in \{0,1\}}\sqrt{p_{x}}\ket{a,x}\\ \otimes \ket{\sqrt{\eta_{AE}}\alpha_{q(a,x)}} \otimes \ket{\sqrt{1-\eta_{AE}}\alpha_{q(a,x)}}.
\end{multline}
We also assume Alice uses a phase randomized source. Let $\phi \in \mathbb{R}$ and $\alpha_{q(a,x)} = \sqrt{\mu}e^{\I\phi}$ for all $a$ and $x$. The resulting phased randomized coherent state takes the form
\begin{equation} \label{eq:rhoprime}
    \begin{aligned}
    \rho'_{ABF} &:= \frac{1}{2\pi}\int_{0}^{2\pi}\mathrm{d}\phi \ \ketbra{\psi_{\mathrm{wcp}}'}{\psi_{\mathrm{wcp}}'} \\ 
    &= \sum_{a,a',x,x' \in \{0,1\}}\sqrt{p_{x}p_{x'}} \ketbra{a,x}{a',x'} \\
    & \hspace{2cm} \otimes \sum_{n=0}^{\infty} P_{\mu}(n) \ketbra{\varphi_{n,a,x}}{\varphi_{n,a',x'}},
    \end{aligned}
\end{equation}
where $P_{\mu}(n) = e^{-\mu}\mu^{n}/n!$ is a Poisson distribution with mean $\mu$,  
\begin{equation}
    \ket{\varphi_{n,a,x}} := \frac{\big( \sqrt{\eta_{AE}}\hat{b}^{\dagger}_{s(a,x)} + \sqrt{1-\eta_{AE}}\hat{f}^{\dagger}_{s(a,x)}\big)^{n} }{\sqrt{n!}}\ket{0},
\end{equation}
and $\big( s(0,0),s(1,0),s(0,1),s(1,1) \big) = (H,V,+,-)$. 

After Eve's action and the second beam splitter $U_{2}$, Bob measures a POVM on mode $B$, $\{N_{b,y}\}$ (see \cref{app:POVM} for details), and we define the rotated measurements $N'_{b,y} = U_{2}^{\dagger}N_{b,y}U_{2}$. Alice's POVM elements are given by $M_{a,x} = \ketbra{a,x}{a,x}$ as before. 

Next, we specify the set of observables and their expected values. Consider the following depolarizing channel acting on the full Fock space $B$, $\mathcal{K}' : \mathcal{D}[\mathcal{H}_{B}] \to \mathcal{D}[\mathcal{H}_{B}]$, where 
\begin{multline}
    \mathcal{K}'(\sigma) = (1-q)\sigma + q\big(\ketbra{1_{B_{H}},0_{B_{V}}}{1_{B_{H}},0_{B_{V}}} \\ + \ketbra{0_{B_{H}},1_{B_{V}}}{0_{B_{H}},1_{B_{V}}}\big)/2
\end{multline} 
and $q \in [0,1]$. Note that $\mathcal{K}'$ mixes the state with a maximally mixed state on the single photon subspace. We define
\begin{equation}
    \rho^{\mathrm{sim}}_{\mathrm{wcp}} := (\mathcal{I}_{A} \otimes \mathcal{K}')[\sigma_{\mathrm{wcp}}],
\end{equation}
where the phase randomized state in the absence of a bypass channel is given by\footnote{We slightly abuse notation here, viewing $\ket{\psi_{\text{wcp}}}$, as defined in \cref{eq:WCPstate}, as a state on $AB$ rather than $ABF$, since $\ket{\psi_{\text{wcp}}}$ has zero weight in $F$. Then $\sigma_{\text{wcp}}$ is a state on $AB$ only.}
\begin{equation}
    \sigma_{\mathrm{wcp}} := \frac{1}{2\pi}\int_{0}^{2\pi} \mathrm{d}\phi \, \ketbra{\psi_{\mathrm{wcp}}}{\psi_{\mathrm{wcp}}}_{AB},    
\end{equation}
and compute the statistics
\begin{equation}
    p^{\mathrm{wcp}}(a,b,x,y) := \tr[(M_{a,x} \otimes N_{b,y})\rho^{\mathrm{sim}}_{\mathrm{wcp}}]. 
\end{equation}
We consider the full set of measurement statistics, so any feasible $\rho \in \mathcal{D}[\mathcal{H}_{ABF}^{\mathrm{wcp}}]$ must satisfy $\tr[(M_{a,x}\otimes N_{b,y}')\rho] = p^{\mathrm{wcp}}(a,b,x,y)$. Note the statistics $p^{\mathrm{wcp}}(a,b,x,y)$ are fixed by the values $q, \, \mu$ and $p_{z}$, and are independent of the bypass parameters. They only serve as a fixed set of observations, which allow us to study the behaviour of the bypass channel in compatible regimes.  

\subsection{Applying dimension reduction}
We consider the following projector for the dimension reduction technique:
\begin{multline}
    \Omega^{\text{wcp}} = \Omega_{0} + \Omega_{1} + \Omega_{2} \\ = \Pi_{B}^{0} \otimes \Pi_{F}^{0} + \Pi_{B}^{0} \otimes \Pi_{F}^{1} +  \Pi_{B}^{1} \otimes \Pi_{F}^{0} \\ + \Pi_{B}^{0} \otimes \Pi_{F}^{2} + \Pi_{B}^{2} \otimes \Pi_{F}^{0} + \Pi_{B}^{1} \otimes \Pi_{F}^{1}.  
\end{multline}
This corresponds to the subspace of less than or equal to 2 photons shared between $B$ and $F$. The resulting space $\Omega^{\text{wcp}}[\mathcal{H}_{BF}^{\text{wcp}}]$ is spanned by 15 vectors, and for our numerics we construct the unitaries $U_{1}, \ U_{2}$ and measurements $N_{b,y}$ in this subspace (recall these operators are block diagonal), allowing us to specify the finite dimensional optimization in \cref{eq:finalOptDR}. We leave the weight $W$ of the state outside $\Omega^{\text{wcp}}[\mathcal{H}_{BF}^{\text{wcp}}]$ as a parameter to choose or bound later (see Appendix \ref{app:weight}).

\subsection{Results}

We set our noise parameter $q=0.02$ and $p_{z} = 0.5$. For a given $\mu$, Bob's detection probability and error rate in each basis are labeled $\text{Q}_{\mu}$, $\text{E}_{\mu}^{Z}$ and $\text{E}_{\mu}^{X}$, respectively. These are calculated from the simulated data:
\begin{align}
    \text{Q}_{\mu} &= 1 - \tr[(\id_{A} \otimes N_{\perp,\perp})\rho_{\mathrm{wcp}}^{\mathrm{sim}}] \\ 
    p_{z}^{2}\text{Q}_{\mu}\text{E}_{\mu}^{Z} &= \sum_{a \neq b}\tr[(M_{a,0} \otimes N_{b,0})\rho_{\mathrm{wcp}}^{\mathrm{sim}}] \\
    (1-p_{z})^{2}\text{Q}_{\mu}\text{E}_{\mu}^{X} &= \sum_{a \neq b}\tr[(M_{a,1} \otimes N_{b,1})\rho_{\mathrm{wcp}}^{\mathrm{sim}}].
\end{align}
In the model considered here, $\text{Q}_{\mu} = 1-(1-q)e^{-\mu}$ and $\text{E}_{\mu}^{Z} = \text{E}_{\mu}^{X}$. The error correction term is given by
\begin{equation}
    \text{Pr}[\text{pass}] H(\mathsf{A}|\mathsf{BXY}) = (1-2p_{z}(1-p_{z}))\mathrm{Q}_{\mu}H_{\text{bin}}(\text{E}_{\mu}).  
\end{equation}
Using the values $\text{Q}_{\mu}$ and $\text{E}_{\mu}$, we can compare our results with the lower bound of~\cite{ghalaii2023satellitebased}:
\begin{multline}
    r^{\infty} \geq (1-2p_{z}(1-p_{z}))\text{Q}_{\mu}\Big[ -H_{\mathrm{bin}}(\text{E}_{\mu}) \\ + \frac{S_{11}^{\mathrm{L}}}{\text{Q}_{\mu}}\big(1-H_{\mathrm{bin}}(\varepsilon_{11}^{\mathrm{U}})\big) + \frac{S_{0}^{\mathrm{L}}}{\text{Q}_{\mu}}\Big], \label{eq:extBP}
\end{multline}
where $S_{0}^{\mathrm{L}} = \max\{\text{Q}_{\mu}-(1-e^{-\mu\eta_{AE}}),0\}$, $S_{11}^{\mathrm{L}} = \max\{\text{Q}_{\mu}-(1-\mu \eta_{AE}e^{-\mu}),0\}$ and $\varepsilon_{11}^{\mathrm{U}} = \min\{\text{E}_{\mu}\text{Q}_{\mu}/S_{11}^{\mathrm{L}},1/2\}$. For our comparisons, we will explore regions where $\mu \in [0.5,1.1]$ and $\eta_{AE} \in [0.85,1]$, which significantly differs to that considered in~\cite{ghalaii2023satellitebased}, where $\mu \approx O(10^{3})$ and $\eta_{AE} \approx O(10^{-3})$. 

The results are plotted in \cref{fig:WCP1_a,fig:WCP1_b,fig:WCP_2}, and are summarized below before elaborated on further.
\begin{itemize}
    \item In \cref{fig:WCP1_a}, we consider the case $\eta_{T} = 1$ and $W = 0$, corresponding to heuristic upper bounds on the secret key rate (see \cref{rem:ub}). We find that for various values of mean photon number, the key rate increases as the restriction on Eve increases (when $\eta_{AE}$ takes smaller values). 
    \item In \cref{fig:WCP1_b}, we minimize the rate over $\eta_{T}$ at a fixed mean photon number for both $W = 0$ and $W > 0$. We also find that the key rate increases as the restriction on Eve increases.
    \item In \cref{fig:WCP_2}, we display the values of $\eta_{T}$ that achieved the minimum in \cref{fig:WCP1_b}, showing an increase with $\eta_{AE}$. 
\end{itemize}

\subsubsection{Heuristic upper bounds}

In \cref{fig:WCP1_a}, we consider heuristic upper bounds on the key rate as a function of $\eta_{AE}$ when $W=0$, which corresponds to a photon squashing assumption: Eve can only interact with and send at most two photons. We then set $\eta_{T} = 1$ (as opposed to minimizing over $\eta_{T}$), and plot the rate for different values of $\mu$. The figure shows the bypass channel having a significant impact on the key rate, suggesting that improved key rates are possible in this regime. Note the interesting relationship between $\mu$ and $\eta_{AE}$. It seems that the optimal $\mu$ increases as $\eta_{AE}$ decreases, which aligns with the observations found in~\cite{ghalaii2023satellitebased} for the regime $\mu \approx O(10^{3})$ and $ \eta_{AE} \approx O(10^{-3})$. As the restriction on Eve increases, it becomes advantageous for Alice to use a higher mean photon number, since she can increase the probability of Bob's detector clicking without becoming vulnerable to photon number splitting attacks. Our results suggest this intuition still holds in the less extreme, more practical regime\footnote{By more practical, we mean in the sense of a smaller restriction on Eve (larger $\eta_{AE}$), which is likely to be achievable by existing monitoring techniques such as LIDAR~\cite{ghalaii2023satellitebased}.} of $\mu \approx O(10^{-1})$, $\eta_{AE} \approx O(10^{-1})$. Note however, our results show this trend for the heuristic upper bound in \cref{fig:WCP1_a}, and further investigation is needed to obtain the same conclusions for lower bounds.

\begin{figure}[h]
    \includegraphics[width=7.5cm]{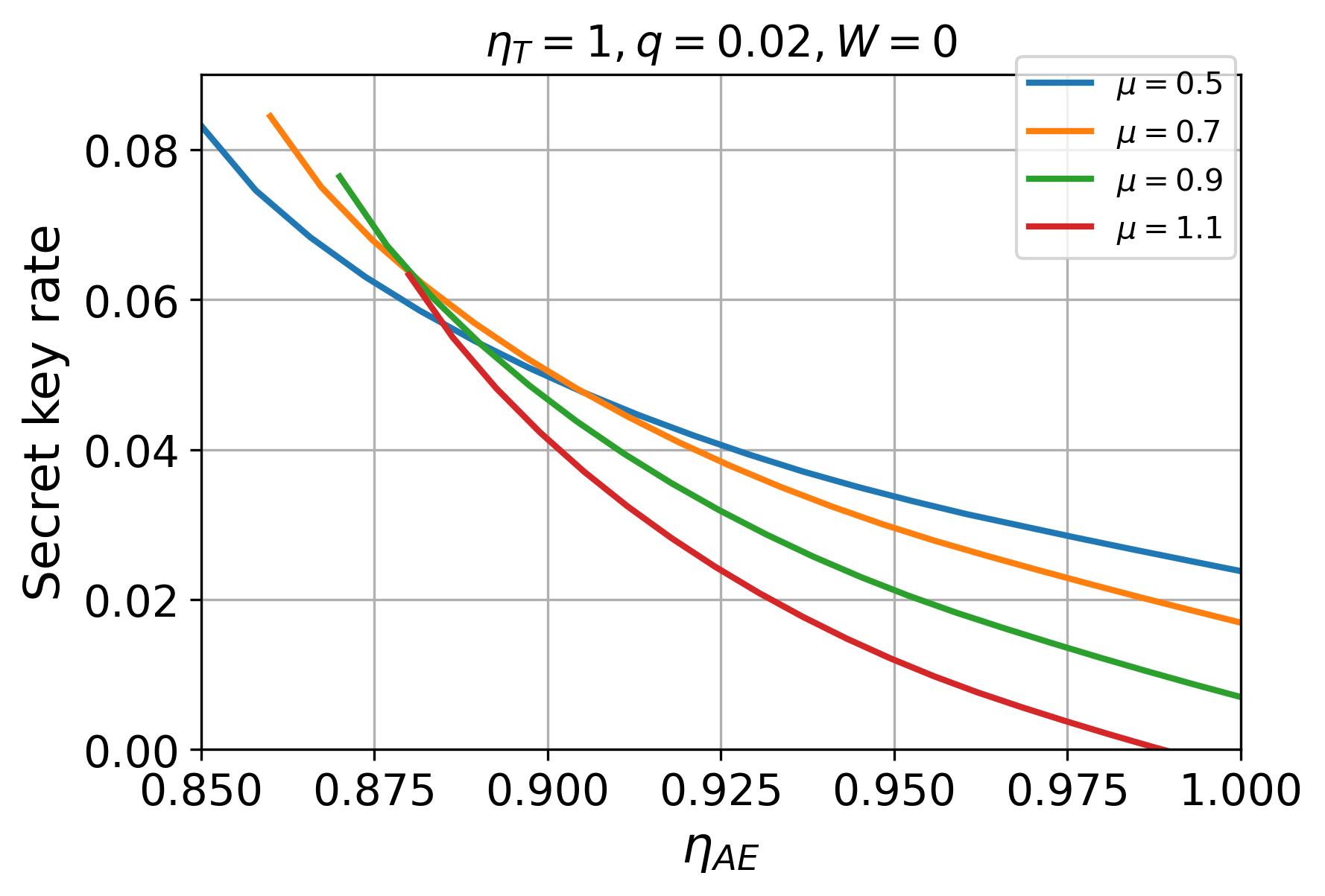}
    \caption{Heuristic upper bounds on the secret key rate of BB84 with phase-randomized weak coherent pulses in the presence of bypass channels. Here, $q$ governs the amount of noise in the system, $\mu$ is the mean photon number and the parameter $W=0$ implies Eve is limited to sending and receiving at most two photons. To obtain the heuristic upper bounds, we set the bypass parameter $\eta_{T} = 1$.}
    \label{fig:WCP1_a}
\end{figure}

\subsubsection{Lower bounds}

Improved key rates from restricting Eve in the bypass model are further supported by \cref{fig:WCP1_b}, where for each value of $\eta_{AE}$, we take the worst case key rate over $\eta_{T}$. Note that this should be done in practice, since unlike $\eta_{AE}$, $\eta_{T}$ cannot be directly measured in the experimental setup~\cite{ghalaii2023satellitebased}. We find that even in this worse case, the key rate improves for smaller values of $\eta_{AE}$, arising from larger restrictions on Eve (cf. blue versus dashed orange line in \cref{fig:WCP1_b}). This improvement holds both for the case of $W=0$ and $W>0$. However, we find the key rate drops sharply with increasing $W$, which could be down to the choice of finite dimensional subspace when applying the dimension reduction technique. This behavior also prevents us from obtaining a non-zero key rate when using the analytical bounds for $W$ derived in \cref{app:weight}. Nonetheless, the larger key rates obtained from including bypass restrictions are encouraging.   

We further compare this to the lower bound of Ref.~\cite{ghalaii2023satellitebased} given by \cref{eq:extBP}. For the regime of $\eta_{AE}$ and $\mu$ we consider, \cref{eq:extBP} does not give a positive key rate, and our results show the potential improvement with our numerical approach. Note however that here we have assumed particular values of $W$, which is not assumed in \cref{eq:extBP}. The improvement is therefore not conclusive at this stage, though appears promising.   

\begin{figure}[h]
    \includegraphics[width=7.5cm]{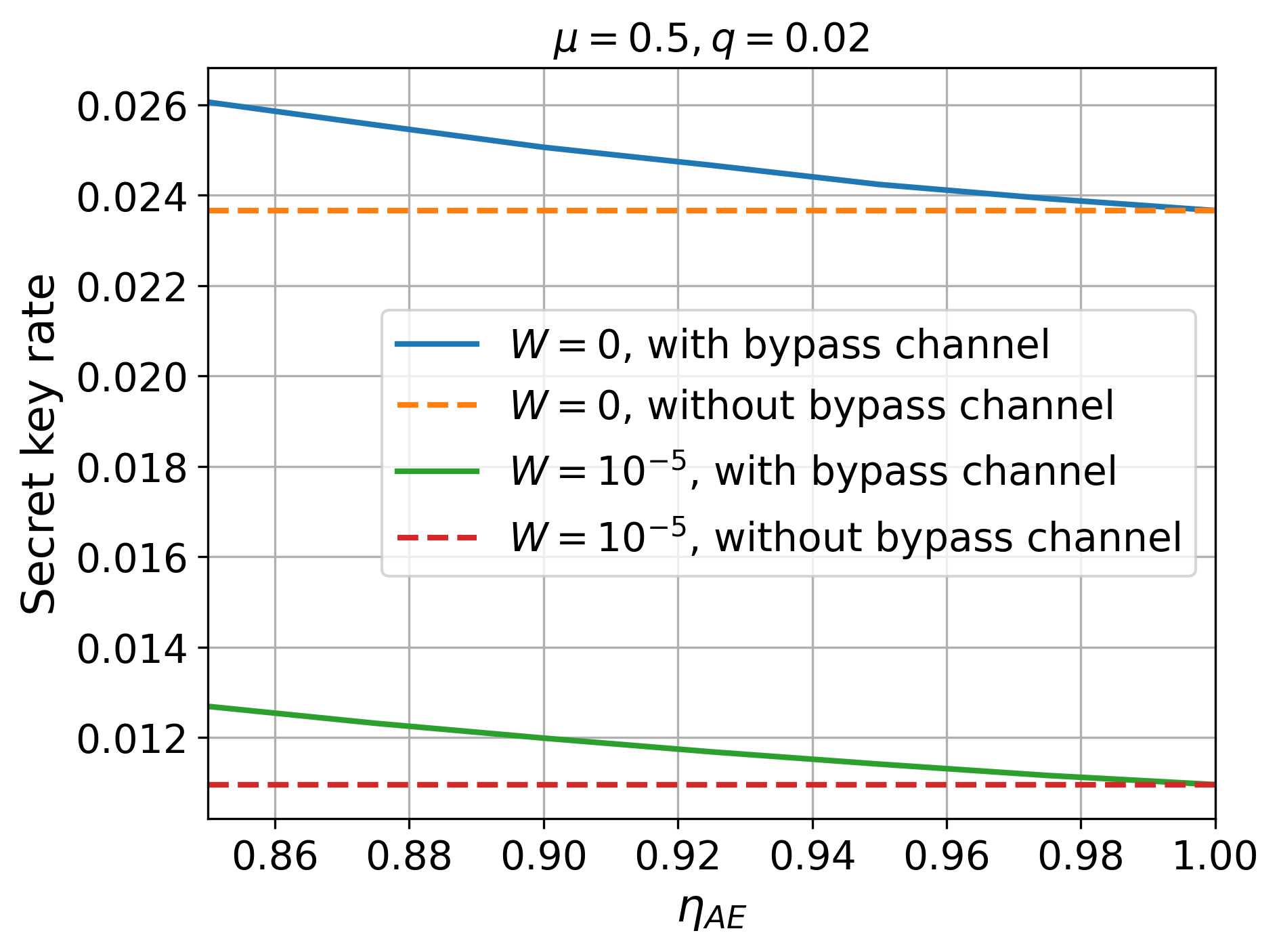}
    \caption{Lower bounds on the key rate of same protocol, obtained by minimizing over $\eta_{T}$ at each $\eta_{AE}$. This is compared to the key rate without a bypass channel, derived from our numerical framework when $\eta_{AE} = \eta_{T} = 1$. We compare both the case of $W=0$ and $W=10^{-5}$, where the analysis bounds the weight of the state outside the two photon subspace by $W$ using the dimension reduction technique. The lower bound derived in Ref.~\cite{ghalaii2023satellitebased} (cf. \cref{eq:extBP}) is negative for all values of $\eta_{AE}$, and hence omitted from the plot. Note however no assumption on the value of $W$ is made in Ref.~\cite{ghalaii2023satellitebased}.}
    \label{fig:WCP1_b}
\end{figure}

In \cref{fig:WCP_2}, we study the behavior of the key rate with $\eta_{T}$ at different values of $\eta_{AE}$. \cref{fig:WCP_2}(a) shows that as $\eta_{AE}$ decreases, the range of $\eta_{T}$ compatible with the observed statistics increases. We also see that for smaller values of $\eta_{AE}$, the lowest key rate is achieved at a smaller value of $\eta_{T}$. This trend is displayed in \cref{fig:WCP_2}(b), where for each $\eta_{AE}$, we plot the $\eta_{T}$ that achieves the minimum. A possible justification for this behavior is that as the restriction on Eve increases, i.e., $\eta_{AE}$ decreases, Eve can access a smaller fraction of the signal sent from Alice. When $\eta_{T}$ is close to 1, the fraction not received by Eve is also not received by Bob, forcing Eve to behave honestly in order to satisfy the observed amount of noise, e.g., achieve the fixed QBER. On the other hand, when $\eta_{T} <1$, Eve can be more dishonest as more of Alice's signal reaches Bob via the bypass channel. If $\eta_{T}$ is too low, there is little transmissivity between Eve and Bob, and her possible attacks become restricted. 

\begin{figure}[h]
    \includegraphics[width=0.9\linewidth]{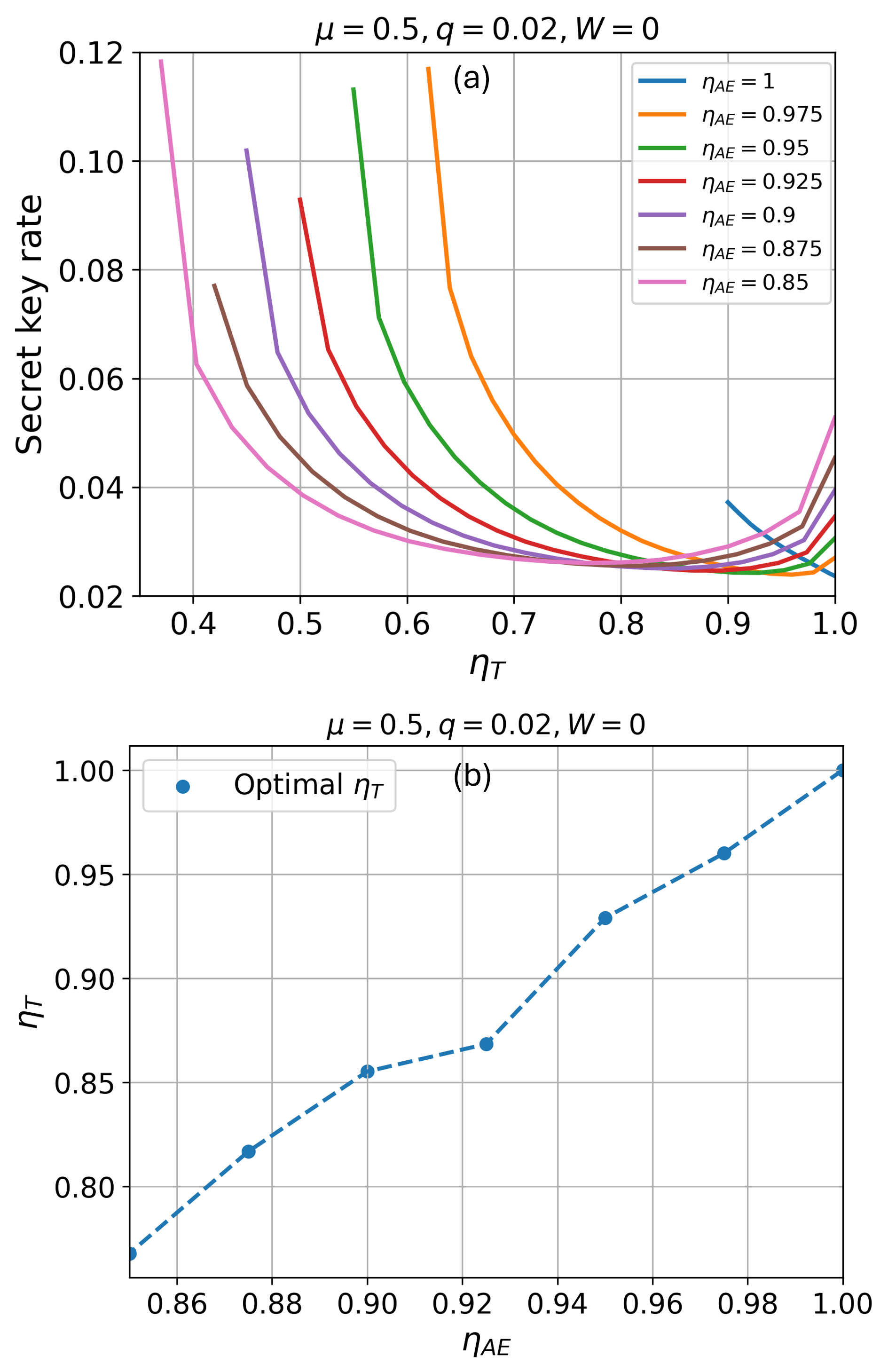}
\caption{The dependence of the key rate on the bypass parameter $\eta_{T}$ for BB84 with weak coherent pulses. (a) The key rate as a function of $\eta_{T}$ for fixed values of $\eta_{AE}$. (b) The value of $\eta_{T}$ that achieves the minimum key rate versus $\eta_{AE}$. The mean photon number, noise parameter and dimension reduction parameter are given by $\mu$, $q$ and $W$, respectively.} 
\label{fig:WCP_2}
\end{figure}

\ifarxiv\section{Summary and Discussion}\label{sec:discussion}\else\noindent{\it Summary and Discussion.|}\fi

Understanding how physically restricted eavesdropping impacts QKD is both an interesting mathematical problem as well as a possible tool for improving key rates in certain contexts, such as satellite based QKD. In this work, we detailed a versatile numerical technique to provide key rates for discrete variable protocols in such settings. We first studied the simplest case of single photon BB84, as was addressed in Ref.~\cite{ghalaii2023satellitebased}. There, it was found that the bound tailored to the bypass channel offered no improvement over the standard key rate expression. Using our numerical approach, we confirmed this, concluding there is indeed no improvement from noiseless bypass channels for the simplest case of single photons. 

Next, we developed the single photon case beyond that considered in~\cite{ghalaii2023satellitebased} by introducing a mismatch in the receiver detector efficiencies. This was made possible due to the versatility of our numerical approach, and the standard analysis (without a bypass channel) was provided by~\cite{Winick_2018}. Contrary to the matched case, we find concrete improvements from including bypass restrictions, namely, both the rate is increased and non-zero key rates can be obtained at higher mismatches. It would be interesting to explore other imperfections and attacks for which the bypass channel benefits beyond the standard case. For example, a Trojan-horse attack~\cite{Vakhitov_01,Gisin06} can be modeled numerically~\cite{Winick_2018}, and can hence be adapted to the bypass channel using our framework. One might expect an improvement here, since any signal sent from Eve to Alice will pass through the first beam-splitter, reducing its intensity, and any reflected signal will experience the same effect. Moreover, a fraction of these reflections may reach Bob via the bypass channel, creating an interesting QKD scenario.        

Finally, we explored the possibility of using numerical techniques to tackle BB84 with phase randomized weak coherent pulses. We explored the practical case of low mean photon number and weak restrictions on Eve, an opposite regime to that previously investigated in Ref.~\cite{ghalaii2023satellitebased}. Our initial findings are encouraging, and we found that moderate restrictions on Eve can increase the key rate when weak coherent pulses are used. However, the current analysis based on dimension reduction is limited to regimes in which the parameter $W$, denoting the probability of the state being outside the finite subspace, is small. This may not be achievable in practice, where $W$ should be estimated directly from the observed statistics such as double-click events~\cite{Upadhyaya2021,Zhang_2021}. Making the bound more tolerant to larger values of $W$ could be achieved by taking a larger finite dimensional subspace when applying the dimension reduction technique. Such an analysis could lead to further improvements from bypass restrictions, though it would demand greater computational resources. We leave this as an avenue for future investigation.      

\ifarxiv\acknowledgments\else\medskip\noindent{\it Acknowledgements}|\fi L.W. is supported by EPSRC via the Quantum Communications Hub (Grant No.\ EP/SO23607/1) and the European Union’s Horizon Europe research and innovation programme under the project “Quantum Secure Networks Partnership” (QSNP, grant agreement No. 101114043). T.U. acknowledges the support of the Natural Sciences and Engineering Research Council of Canada (NSERC) through the Doctoral Postgraduate Scholarship. M.R. acknowledges partial funding from the UK EPSRC grants EP/Y037421/1 and EP/X040518/1.

\onecolumngrid
\appendix
\section{Proof of \cref{lem:relEnt}}
\label{app:proof1}
For this proof, we will need the following theorem that we reproduce for the convenience of the reader.
\begin{theorem}[\cite{Coles12} Theorem 1]
    Let $\ket{\Psi}_{ABE} \in \mathcal{H}_{A} \otimes \mathcal{H}_{B} \otimes \mathcal{H}_{E}$ be a tripartite pure state, with $\mathcal{H}_{A} = \mathbb{C}^{d}$ for a positive integer $d$. Let $\mathcal{Z}:\mathcal{D}(\mathcal{H}_{A}) \to \mathcal{D}(\mathcal{H}_{\mathsf{A}})$ be a pinching quantum channel in the eigenbasis of $\rho_{A} = \tr_{BE}[\ketbra{\Psi}{\Psi}]$, $\mathcal{Z}[\sigma_{A}] = \sum_{a=0}^{d-1} \bra{a}\sigma_{A}\ket{a}\ketbra{a}{a}_{\mathsf{A}}$, and $\tau_{\mathsf{A}BE} = (\mathcal{Z} \otimes \mathcal{I}_{BE})[\ketbra{\Psi}{\Psi}]$. Then
    \begin{equation}
        H(\mathsf{A}|E)_{\tau_{\mathsf{A}E}} = D\big(\tau_{AB} \| (\mathcal{Z} \otimes \mathcal{I}_{B})[\tau_{AB}]\big),
    \end{equation}
    where $\tau_{AB} = \tr_{E}[\ketbra{\Psi}{\Psi}]$. \label{thm:discord}
\end{theorem}
\begin{proof}
    First note that, since $\tau_{ABE} = \ketbra{\Psi}{\Psi}$ is a pure state,
    \begin{equation}
        \tr[\tau_{AB}\log(\tau_{AB})] = -S(\tau_{E}), \label{eq:rel1}
    \end{equation}
    where $S(\rho) = -\tr[\rho \log(\rho)]$ is the von Neumann entropy. We also have
    \begin{equation}
        \log \Big((\mathcal{Z} \otimes \mathcal{I}_{B})[\tau_{AB}]\Big) = \sum_{a=0}^{d-1} \ketbra{a}{a}_{\mathsf{A}} \otimes \log (\tau_{B}^{a})
    \end{equation}
    where $\tau_{B}^{a} = \tr_{A}[(\ketbra{a}{a} \otimes \id_{B})\tau_{AB}]$. This implies 
    \begin{equation}
        \tr\Big[\tau_{AB}\log \Big((\mathcal{Z} \otimes \mathcal{I}_{BE})[\tau_{AB}]\Big)\Big] = \sum_{a=0}^{d-1} \tr[\tau_{B}^{a}\log(\tau_{B}^{a})].
    \end{equation}
    Consider the Schmidt decomposition for $\ket{\Psi}$, $\ket{\Psi} = \sum_{a}\sqrt{p_{a}} \ket{a}_{A} \otimes \ket{\psi_{a}}_{BE}$, where $\{p_{a}\}_{a=0}^{d-1}$ is a probability distribution and $\{\ket{\psi_{a}}\}_{a=0}^{d-1}$ is an set of orthogonal states in $\mathcal{H}_{B} \otimes \mathcal{H}_{E}$. It follows that $\tau_{B}^{a} = p_{a} \tr_{E}[\ketbra{\psi_{a}}{\psi_{a}}]$, and
    \begin{equation}
        \sum_{a=0}^{d-1}\tr[\tau_{B}^{a}\log(\tau_{B}^{a})] = -H(\{p_{a}\}_{a=0}^{d-1}) -\sum_{a=0}^{d-1}p_{a}S(\hat{\tau}_{B}^{a}) = -H(\{p_{a}\}_{a=0}^{d-1}) -\sum_{a=0}^{d-1}p_{a}S(\hat{\tau}_{E}^{a}) = -S(\tau_{\mathsf{A}E}), \label{eq:rel2}
    \end{equation}
    where $\hat{\tau}_{B}^{a} = \tr_{E}[\ketbra{\psi_{a}}{\psi_{a}}]$, $\hat{\tau}_{E}^{a} = \tr_{B}[\ketbra{\psi_{a}}{\psi_{a}}]$ and $H(\{p_{a}\}_{a=0}^{d-1}) = -\sum_{a}p_{a}\log(p_{a})$ is the Shannon entropy. The first equality follows from a direct calculation. To obtain the second inequality, we used the fact that $S(\hat{\tau}_{B}^{a}) = S(\hat{\tau}_{E}^{a})$ since $\ketbra{\psi_{a}}{\psi_{a}}_{BE}$ is pure. The final equality follows from the fact that 
    \begin{equation}
        \tau_{\mathsf{A}E} = \tr_{B}[(\mathcal{Z}\otimes \mathcal{I})[\ketbra{\Psi}{\Psi}]] = \sum_{a}\ketbra{a}{a}_{\mathsf{A}}\otimes p_{a}\tr_{B}[\ketbra{\psi_{a}}{\psi_{a}}] = \sum_{a}\ketbra{a}{a}_{\mathsf{A}}\otimes \tau_{E}^{a}
    \end{equation}
    which implies
    \begin{equation}
        S(\tau_{\mathsf{A}E}) = - \sum_{a=0}^{d-1}\tr[\tau_{E}^{a}\log(\tau_{E}^{a})] = H(\{p_{a}\}_{a=0}^{d-1}) +\sum_{a=0}^{d-1}p_{a}S(\hat{\tau}_{E}^{a}).
    \end{equation}
    Combining \cref{eq:rel1,eq:rel2} with the definition of the conditional von Neumann entropy, we find
    \begin{equation}
         \tr[\tau_{AB}\log(\tau_{AB})] - \tr\Big[\tau_{AB}\log \Big((\mathcal{Z} \otimes \mathcal{I}_{BE})[\tau_{AB}]\Big)\Big] = 
        - S(\tau_{E}) + S(\tau_{\mathsf{A}E}) = H(\mathsf{A}|E)_{\tau_{\mathsf{A}E}}.
    \end{equation}
    Noting that the above expression is always finite, the left hand side is exactly equal to $D\big(\tau_{AB} \| (\mathcal{Z} \otimes \mathcal{I}_{B})[\tau_{AB}]\big)$, completing the proof. 
\end{proof}

We now prove the main result of this appendix. 

\vspace{0.2cm}

\noindent \textbf{Lemma 1.} \textit{Let $\big (\ket{\psi^{\mathrm{init}}},U_{1},U_{BE},U_{2}, M_{A},N_{B} \big)$ be any quantum strategy in the bypass scenario that gives rise to the post-selected state $\tilde{\rho}$ in \cref{eq:pms}. Then the following holds:
\begin{equation*}
    \mathrm{Pr}[\mathrm{pass}] \, H(\mathsf{A}|\mathsf{XY}E)_{\tilde{\rho}} \\ = \sum_{x \in \{0,1\}} D\Big( V_{x}\rho_{ABF,x}'V_{x}^{\dagger} \big \| (\mathcal{Z}_{\tilde{A}} \otimes \mathcal{I}_{ABF}) \big[ V_{x}\rho_{ABF,x}'V_{x}^{\dagger} \big] \Big),
\end{equation*}
where
\begin{equation*}
    \rho_{ABF,x}' = \tr_{E}\big[\ketbra{\psi_{x}'}{\psi_{x}'}\big], \ \ \ \ \ket{\psi_{x}'} = \Bigg(\id_{A} \otimes U_{2}^{\dagger}\sqrt{\sum_{b \in \{0,1\}}N_{b,x}}U_{2} \otimes \id_{E}\Bigg)\ket{\psi'},
\end{equation*}
$V_{x}: \mathcal{H}_{A}\otimes \mathcal{H}_{BF} \to \mathcal{H}_{\tilde{A}}\otimes \mathcal{H}_{A}\otimes \mathcal{H}_{BF}$, where $\mathcal{H}_{\tilde{A}} \cong \mathbb{C}^{2}$, is defined by
\begin{equation*}
    V_{x} = \sum_{a \in \{0,1\}} \ket{a}_{\tilde{A}} \otimes M_{a,x} \otimes \id_{BF}
\end{equation*}
and satisfies $\sum_{x}V_{x}^{\dagger}V_{x} = \id_{ABF}$, and $\mathcal{Z}_{\tilde{A}}:\mathcal{L}(\mathcal{H}_{\tilde{A}}) \to \mathcal{L}(\mathcal{H}_{\mathsf{A}})$ is the pinching channel on $\tilde{A}$,
\begin{equation*}
    \mathcal{Z}_{\tilde{A}}[\sigma] = \sum_{a \in \{0,1\}} \bra{a}\sigma \ket{a}\ketbra{a}{a}_{\mathsf{A}}
\end{equation*}
for all $\sigma \in \mathcal{L}(\mathcal{H}_{\tilde{A}})$.} 

\begin{proof}
    Recall that the post-measurement state $\tilde{\rho}$ is given by
    \begin{equation}
        \tilde{\rho}_{\mathsf{ABXY}E}  
    = \frac{1}{\text{Pr}[\text{pass}]} \sum_{x \in \{0,1\}} \sum_{a,b \in \{0,1\}} \ketbra{a,b}{a,b}_{\mathsf{AB}} \\ \otimes \ketbra{x,x}{x,x}_{\mathsf{XY}} \otimes \rho_{E}^{a,b,x,x}.
    \end{equation}
    with $\rho_{E}^{a,b,x,y} = \tr_{ABF}\big[ (M_{a,x} \otimes N_{b,y} \otimes \id_{E})\ketbra{\psi}{\psi}_{ABFE}\big]$. We begin by noting that, for classical quantum states, the relative entropy conditioned on a classical system takes the form 
\begin{equation}
    H(\mathsf{A}|\mathsf{XY}E)_{\tilde{\rho}} = \sum_{x,y} \text{Pr}[\mathsf{X}=x,\mathsf{Y}=y|\text{pass}]H(\mathsf{A}|\mathsf{X}=x,\mathsf{Y}=y,E)_{\tilde{\rho}} = \sum_{x} \text{Pr}[\mathsf{X} = x,\mathsf{Y}=x|\text{pass}] H(\mathsf{A}|E)_{\tilde{\rho}_{\mathsf{A}E|\mathsf{X}=\mathsf{Y}=x}}, \label{eq:cdecomp}
\end{equation}
where
\begin{equation}
\begin{gathered}
    \tilde{\rho}_{\mathsf{A}E|\mathsf{X}=\mathsf{Y}=x} = \frac{1}{\text{Pr}[\mathsf{X} = x,\mathsf{Y}=x|\text{pass}]}\frac{1}{\text{Pr}[\text{pass}]} \sum_{a} \ketbra{a}{a}_{\mathsf{A}} \otimes \rho_{E}^{a,x} \\
    \text{Pr}[\mathsf{X} = x,\mathsf{Y}=x|\text{pass}] = \frac{1}{\text{Pr}[\text{pass}]} \sum_{a,b} \tr[\rho_{E}^{a,b,x,x}] \ \ \ \  \text{and} \\
    \rho_{E}^{a,x} = \sum_{b}\rho_{E}^{a,b,x,x} = \tr_{ABF} \big[ (M_{a,x} \otimes \sum_{b}N_{b,x}\otimes \id_{FE})\ketbra{\psi}{\psi}\big].
\end{gathered}
\end{equation}
Using the cyclic property of the partial trace we can write $\rho_{E}^{a,x} = \tr_{ABF} \big[ (M_{a,x} \otimes \id_{BFE})\ketbra{\psi_{x}}{\psi_{x}}\big]$ where
\begin{equation}
    \ket{\psi_{x}} = \Bigg(\id_{A} \otimes \sqrt{\sum_{b} N_{b,x}} \otimes \id_{E}\Bigg)\ket{\psi}.
\end{equation}
By noting $\text{Pr}[\mathsf{X} = x, \mathsf{Y}= x] = \text{Pr}[\text{pass}]\text{Pr}[\mathsf{X} = x,\mathsf{Y}=x|\text{pass}]$, we can put this together to obtain 
\begin{equation}
     \tilde{\rho}_{\mathsf{A}E|\mathsf{X}=\mathsf{Y}=x} = \frac{1}{\text{Pr}[\mathsf{X} = x,\mathsf{Y}=x]} \sum_{a} \ketbra{a}{a}_{\mathsf{A}} \otimes \tr_{ABF} \big[ (M_{a,x} \otimes \id_{BFE})\ketbra{\psi_{x}}{\psi_{x}}\big].
\end{equation}

Having decomposed the objective function, we now rewrite it in terms of the quantum relative entropy. Define the map $V_{x}:\mathcal{H}_{A} \otimes \mathcal{H}_{BF}  \to \mathcal{H}_{\tilde{A}} \otimes \mathcal{H}_{A} \otimes \mathcal{H}_{BF}$,
\begin{equation}
    V_{x} = \sum_{a} \ket{a}_{\tilde{A}} \otimes M_{a,x} \otimes \id_{BF}.
\end{equation}
Note that $(V^{\dagger}_{x}V_{x} \otimes \id_{E})\ket{\psi_{x}} = \ket{\psi_{x}}$. We then have
\begin{equation}
    (V_{x} \otimes \id_{E})\ketbra{\psi_{x}}{\psi_{x}}(V^{\dagger}_{x} \otimes \id_{E}) = \sum_{a,a'}\ketbra{a}{a'}_{\tilde{A}} \otimes (M_{a,x}\otimes \id_{BFE})\ketbra{\psi_{x}}{\psi_{x}}(M_{a',x}\otimes \id_{BFE}).
\end{equation}
Next we define the pinching channel and $\mathcal{Z}_{\tilde{A}}:\mathcal{L}(\mathcal{H}_{\tilde{A}}) \to \mathcal{L}(\mathcal{H}_{\mathsf{A}})$ is the pinching channel on $\tilde{A}$,
\begin{equation}
    \mathcal{Z}_{\tilde{A}}[\sigma] = \sum_{a \in \{0,1\}} \bra{a}\sigma \ket{a}\ketbra{a}{a}_{\mathsf{A}},
\end{equation}
and define
\begin{equation}
    \tau_{\mathsf{A}E|x} := \frac{1}{\text{Pr}[\mathsf{X} = x,\mathsf{Y}=x]}\Big( \mathcal{Z}_{\tilde{A}} \otimes \mathcal{I}_{E} \Big)\Big[\tr_{ABF}[(V_{x} \otimes \id_{E})\ketbra{\psi_{x}}{\psi_{x}}(V^{\dagger}_{x} \otimes \id_{E})]\Big] .
\end{equation}
By direct calculation
\begin{multline}
    \frac{1}{\text{Pr}[\mathsf{X} = x,\mathsf{Y}=x]}\Big( \mathcal{Z}_{\tilde{A}} \otimes \mathcal{I}_{E} \Big)\Big[\tr_{ABF}[(V_{x} \otimes \id_{E})\ketbra{\psi_{x}}{\psi_{x}}(V^{\dagger}_{x} \otimes \id_{E})]\Big] \\ = \frac{1}{\text{Pr}[\mathsf{X} = x,\mathsf{Y}=x]}\sum_{a} \ketbra{a}{a}_{\mathsf{A}} \otimes  \tr_{ABF}[(M_{a,x}\otimes \id_{BFE})\ketbra{\psi_{x}}] = \tilde{\rho}_{\mathsf{A}E|\mathsf{X}=\mathsf{Y}=x},
\end{multline}
and therefore
\begin{equation}
    H(\mathsf{A}|E)_{\tilde{\rho}_{\mathsf{A}E|\mathsf{X}=\mathsf{Y}=x}} = H(\mathsf{A}|E)_{\tau_{\mathsf{A}E|x}}.
\end{equation}
Note that $\ket{\Psi_{x}} := (V_{x} \otimes \id_{E})\ket{\psi_{x}} / \sqrt{\text{Pr}[\mathsf{X}=x,\mathsf{Y}=x]}$ is a normalized pure state on systems $\tilde{A}ABFE$, and $\mathcal{Z}$ is a pinching quantum channel which maps $\tilde{A} \mapsto \mathsf{A}$. Let us write 
\begin{equation}
    \tau_{\tilde{A}ABF|x} = \tr_{E}[\ketbra{\Psi_{x}}{\Psi_{x}}].
\end{equation}
Then applying \cref{thm:discord} with $\ket{\Psi} = \ket{\Psi_{x}}$, 
\begin{equation}
    H(\mathsf{A}|E)_{\tau_{\mathsf{A}E|x}} = D\Big( \tau_{\tilde{A}ABF|x} \| (\mathcal{Z}_{\tilde{A}} \otimes \mathcal{I}_{ABF})[\tau_{\tilde{A}ABF|x}] \Big) = \frac{1}{\text{Pr}[\mathsf{X}=x,\mathsf{Y}=x]}D\Big( V_{x}\rho_{ABF,x}V_{x}^{\dagger} \| (\mathcal{Z}_{\tilde{A}} \otimes \mathcal{I}_{ABF})[V_{x}\rho_{ABF,x}V_{x}^{\dagger}] \Big),
\end{equation}
where we used the fact that $V_{x}$ and $\tr_{E}$ commute, defined $\rho_{ABF,x} := \tr_{E}[\ketbra{\psi_{x}}{\psi_{x}}]$ and we used the identity $c D(\rho \| \sigma) = D(c \rho \| c \sigma)$ for a constant $c > 0$. Combining with \cref{eq:cdecomp} we obtain
\begin{equation}
    \text{Pr}[\text{pass}] \, H(\mathsf{A}|\mathsf{XY}E)_{\tilde{\rho}} = \sum_{x} D\Big( V_{x}\rho_{ABF,x}V_{x}^{\dagger} \| (\mathcal{Z}_{\tilde{A}} \otimes \mathcal{I}_{ABF})[V_{x}\rho_{ABF,x}V_{x}^{\dagger}] \Big),
\end{equation}
where we used the fact that $\text{Pr}[\mathsf{X}=x,\mathsf{Y}=y] = \text{Pr}[\text{pass}] \, \text{Pr}[\mathsf{X}=x,\mathsf{Y}=y | \text{pass}]$.

To complete the proof, we need to rewrite the objective function in terms of the state prior to the final beam splitter. To do so, we will use the fact that $D(\rho \| \sigma) = D(U\rho U^{\dagger} \| U\sigma U^{\dagger})$ for any unitary $U$. Recall $U_{2}$ is the unitary on $BF$ that models the final beam splitter, and satisfies 
\begin{equation}
    (\id_{A} \otimes U_{2} \otimes \id_{E})\ket{\psi'}_{ABFE} = \ket{\psi}_{ABFE}.
\end{equation}
Suppressing the tensor product with identity, $U_{2} \equiv \id_{A} \otimes U_{2} \otimes \id_{E}$, we have 
\begin{equation}
\begin{aligned}
    D\Big( V_{x}\rho_{ABF,x}V_{x}^{\dagger} \| (\mathcal{Z}_{\tilde{A}} \otimes \mathcal{I}_{ABF})[V_{x}\rho_{ABF,x}V_{x}^{\dagger}] \Big) &= D\Big( U_{2}^{\dagger}V_{x}\rho_{ABF,x}V_{x}^{\dagger}U_{2} \| U_{2}^{\dagger}(\mathcal{Z}_{\tilde{A}} \otimes \mathcal{I}_{ABF})[V_{x}\rho_{ABF,x}V_{x}^{\dagger}] U_{2}\Big) \\
    &= D\Big( V_{x}U_{2}^{\dagger}\rho_{ABF,x}U_{2}V_{x}^{\dagger} \| (\mathcal{Z}_{\tilde{A}} \otimes \mathcal{I}_{ABF})[V_{x}U_{2}^{\dagger}\rho_{ABF,x}U_{2}V_{x}^{\dagger}] \Big), 
\end{aligned}
\end{equation}
where we used the fact that $U_{2}$ commutes with $V_{x}$ and $\mathcal{Z}_{\tilde{A}}$. Note that
\begin{equation}
    U_{2}^{\dagger}\rho_{ABF,x}U_{2} = U_{2}^{\dagger}\tr_{E}[\ketbra{\psi_{x}}{\psi_{x}}]U_{2} = \tr_{E}\big[ U^{\dagger}_{2}\ketbra{\psi_{x}}{\psi_{x}}U_{2} \big],
\end{equation}
and recall $\ket{\psi_{x}} = (\id_{A} \otimes \sqrt{\sum_{b}N_{b,x}} \otimes \id_{E})\ket{\psi}$, which implies
\begin{equation}
   U^{\dagger}_{2}\ket{\psi_{x}} = U_{2}^{\dagger}\Bigg(\id_{A} \otimes \sqrt{\sum_{b}N_{b,x}} \otimes \id_{E}\Bigg)U_{2}U_{2}^{\dagger}\ket{\psi} =  (\id_{A} \otimes P_{x} \otimes \id_{E})\ket{\psi'} =: \ket{\psi_{x}'},
\end{equation}
where $P_{y} = U_{2}^{\dagger}\sqrt{\sum_{b}N_{b,y}}U_{2}$. By writing
\begin{equation}
    \rho_{ABF,x}' = \tr_{E}\big[ \ketbra{\psi_{x}'}{\psi_{x}'}\big],
\end{equation}
we find
\begin{equation}
    D\Big( V_{x}\rho_{ABF,x}V_{x}^{\dagger} \| (\mathcal{Z} \otimes \mathcal{I}_{ABF})[V_{x}\rho_{ABF,x}V_{x}^{\dagger}] \Big) = D\Big( V_{x}\rho_{ABF,x}'V_{x}^{\dagger} \| (\mathcal{Z} \otimes \mathcal{I}_{ABF})[V_{x}\rho_{ABF,x}'V_{x}^{\dagger}] \Big),
\end{equation}
completing the proof. 
\end{proof}

\section{Bob's measurements} \label{app:POVM}
In this appendix, we describe Bob's measurement on system $B$ in terms of the POVM elements $N_{b,y}$. These are described as operators acting on the (infinite dimensional) photon number space $\mathcal{H}_{B} \otimes \mathcal{H}_{F}$ (acting trivially on $F$)~\cite{Zhang_2017},
\begin{equation}
\begin{aligned}
    N_{0,0} &= p_{z}\sum_{k,l=0}^{\infty}\sum_{n=1}^{\infty} \ketbra{n_{B_{H}},0_{B_{V}},k_{F_{H}},l_{F_{V}}}{n_{B_{H}},0_{B_{V}},k_{F_{H}},l_{F_{V}}} + \frac{1}{2}N_{HV}, \\
    N_{1,0} &= p_{z}\sum_{k,l=0}^{\infty}\sum_{n=1}^{\infty} \ketbra{0_{B_{H}},n_{B_{V}},k_{F_{H}},l_{F_{V}}}{0_{B_{H}},n_{B_{V}},k_{F_{H}},l_{F_{V}}} + \frac{1}{2}N_{HV},  \\
    N_{HV} &= p_{z}\sum_{k,l=0}^{\infty}\sum_{n,m=1}^{\infty}\ketbra{n_{B_{H}},m_{B_{V}},k_{F_{H}},l_{F_{V}}}{n_{B_{H}},m_{B_{V}},k_{F_{H}},l_{F_{V}}}, \\
    N_{0,1} &= (1-p_{z})\sum_{k,l=0}^{\infty}\sum_{n=1}^{\infty} \ketbra{n_{B_{+}},0_{B_{-}},k_{F_{H}},l_{F_{V}}}{n_{B_{+}},0_{B_{-}},k_{F_{H}},l_{F_{V}}} + \frac{1}{2}N_{\pm}, \\
    N_{1,1} &= (1-p_{z})\sum_{k,l=0}^{\infty}\sum_{n=1}^{\infty} \ketbra{0_{B_{+}},n_{B_{-}},k_{F_{H}},l_{F_{V}}}{0_{B_{+}},n_{B_{-}},k_{F_{H}},l_{F_{V}}} + \frac{1}{2}N_{\pm}, \\
    N_{\pm} &= p_{z}\sum_{k,l=0}^{\infty}\sum_{n,m=1}^{\infty} \ketbra{n_{B_{+}},m_{B_{-}},k_{F_{H}},l_{F_{V}}}{n_{B_{+}},m_{B_{-}},k_{F_{H}},l_{F_{V}}}, \ \ \text{and} \\
    N_{\perp,\perp} &= \id_{BF} - \sum_{b,y \in \{0,1\}}N_{b,y},
\end{aligned}
\end{equation}
where $p_{z}$ is the probability Bob chooses the $H/V$ basis. Notice we sum over the bypass modes $F_{H},F_{V}$ to account for the fact that these degrees of freedom are not accessible to Bob's measurement device. The rotated measurements are defined by $N'_{b,y} = U_{2}^{\dagger}N_{b,y}U_{2}$, where $U_{2}$ is the unitary describing the final beam splitter. We remark that each of the POVM elements above are block diagonal in photon number space. Specifically, they are of the form
\begin{multline}
    N_{b,y} = \tilde{N}_{b,y} \otimes \id_{F} = \Bigg(\sum_{n=0}^{\infty} \Pi_{B}^{n}\tilde{N}_{b,y}\Pi_{B}^{n}\Bigg) \otimes \Bigg(\sum_{m = 0}^{\infty} \Pi_{F}^{m}\Bigg) = \sum_{n,m=0}^{\infty}   \Pi_{B}^{n}\tilde{N}_{b,y}\Pi_{B}^{n} \otimes \Pi_{F}^{m} \\ = \sum_{n=0}^{\infty} \Bigg(\sum_{k=0}^{n} (\Pi_{B}^{n-k} \otimes  \Pi_{F}^{k})N_{b,y}(\Pi_{B}^{n-k} \otimes  \Pi_{F}^{k})\Bigg) = \sum_{n=0}^{\infty} N_{b,y}^{n}, 
\end{multline}
where $\tilde{N}_{b,y}$ is the factor of $N_{b,y}$ that acts on $B$ only ($\tilde{N}_{b,y}$ has the property that it is block diagonal with respect to the total photon number in $B$), and $N_{b,y}^{n} = \sum_{k=0}^{n} (\Pi_{B}^{n-k} \otimes  \Pi_{F}^{k})N_{b,y}(\Pi_{B}^{n-k} \otimes  \Pi_{F}^{k})$ is the block of $N_{b,y}$ acting on the $n$ photon subspace of $\mathcal{H}_{B} \otimes \mathcal{H}_{F}$.

\section{Additional details on the dimension reduction technique} \label{app:pTrace}

In this appendix, we elaborate on the final optimization problem in \cref{eq:finalOptDR} obtained after applying the dimension reduction technique. Specifically, we address how to define the partial trace constraint, and why this is non-trivial in the bypass scenario. 

The quantity $\tr_B[\rho_{ABF}]$ arises in the constraints of \cref{eq:finalOptDR}, where the partial trace mapping typically takes a tensor product of two Hilbert spaces to a single one. However, the finite-dimensional Hilbert space  $\Omega [\mathcal{H}_{BF}]$ will not in general be of tensor product form (see the definition of $\mathcal{H}_{BF}^{\mathrm{sp}}$ in \cref{eq:SPH} for a concrete example). Nevertheless, $\tr_B[\rho_{ABF}]$ is still well defined, since we can view the finite dimensional state $\rho_{ABF} \in \mathcal{P}(\mathcal{H}_{A} \otimes \Omega [\mathcal{H}_{BF}])$ as an operator on the infinite dimensional Hilbert space $\mathcal{H}_{A} \otimes \mathcal{H}_{BF}$. Formally, computing the partial trace in this space would require an infinite summation. However, in this appendix, we show that because the operator $\rho_{ABF} \in \mathcal{P}(\mathcal{H}_{A} \otimes \Omega [\mathcal{H}_{BF}])$ is finite dimensional, the partial trace can be computed exactly using only a finite summation.

To see this, let $\ket{i}_{B}$ be an orthonormal basis for $\mathcal{H}_{B}$ and $\ket{j}_{F}$ be an orthonormal basis for $\mathcal{H}_{F}$. Let $\Omega [\mathcal{H}_{BF}]=\mathrm{Span}\{ \ket{ij}_{BF} : {(i,j)}\in S \}$ for some finite set of indices $S$. Next, we partition $S$ by all the possible values of $j$. That is, define  $S_j := \{ i \ :  \ (i,j) \in S \}$. Recall that the partial trace can written in terms of the matrix elements $\bra{k}_F \tr_B[\rho]\ket{l}_F =\sum_i \braket{i,k|\rho}{i,l}$. For any $\rho \in \mathcal{P}(\Omega [\mathcal{H}_{BF}])$, there are only finitely many nonzero contributions to this sum, $\bra{k}_F \tr_B[\rho] \ket{l}_F =\sum_{i \in S_k \cap S_l} \braket{i,k|\rho}{i,l}$. Numerically, can we apply the preceding formula to every block over the $A$ system to compute $\tr_B[\rho_{ABF}]$.
In a similar manner to~\cite[Appendix B.2]{Upadhyaya2021}, we can compute the adjoint of the partial trace map used in the dual key rate SDP. 

In more detail, recall the Hilbert space structure $\mathcal{H}_{A} \otimes \mathcal{H}_{B} \otimes \mathcal{H}_{F}$, where $\mathcal{H}_{A} \cong \mathbb{C}^{4}$, $\mathcal{H}_{B} = \mathcal{H}_{B_{H}} \otimes \mathcal{H}_{B_{F}}$, $\mathcal{H}_{F} = \mathcal{H}_{F_{H}} \otimes \mathcal{H}_{F_{V}}$, and $\mathcal{H}_{B_{H}} \cong \mathcal{H}_{B_{V}} \cong \mathcal{H}_{F_{H}} \cong \mathcal{H}_{F_{V}} \cong \text{Span}\{ \ket{n} \ : \ n = 0,1,2,...\}$. Consider any finite dimensional subspace $\mathcal{H}'_{BF} \subset \mathcal{H}_{B} \otimes \mathcal{H}_{F}$ of dimension $d$. Let $\mathcal{B}_{BF}' = \{\ket{s}\}_{s=0}^{d-1}$ be a basis for $\mathcal{H}_{BF}'$ - note that $\mathcal{H}_{BF}'$ need not decompose into a tensor product over $B$ and $F$. Suppose we choose our subspace and basis such that we can identify basis vectors in the following way:
\begin{equation}
    \ket{s} = \ket{b_{H}^{s},b_{V}^{s},f_{H}^{s},f_{V}^{s}},
\end{equation}
for some non-negative integers $b_{H}^{s},b_{V}^{s},f_{H}^{s},f_{V}^{s}$. Consider a state $\rho_{ABF} \in \mathcal{D}\big(\mathcal{H}_{A} \otimes \mathcal{H}_{BF}'\big)$. We can always write
\begin{equation}
    \rho_{ABF} = \sum_{i,j=0}^{3} \ketbra{i}{j}_{A} \otimes M_{i,j},
\end{equation}
where $M_{i,j}$ is a linear operator on $\mathcal{H}_{BF}'$, i.e.,
\begin{equation}
    M_{i,j} = \sum_{s,s'=0}^{d-1} m_{i,j}^{s,s'} \ketbra{s}{s'}.
\end{equation}
Now we wish to trace out system $B$ from $\rho_{ABF}$. The subspace of $F$ for which $\rho_{F}$ has support is given by
\begin{equation}
    \mathcal{B}'_{F} = \big\{ \ket{f_{H},f_{V}} \ : \ \exists s \in \{0,...,d-1\}  \ \text{s.t.} \ \ket{b_{H}^{s},b_{V}^{s},f_{H},f_{V}} \in \mathcal{B}'_{BF}\big\}.
\end{equation}
We therefore see $\rho_{AF} \in \mathcal{D}(\mathcal{H}_{A} \otimes \mathcal{H}_{F}')$ where $\mathcal{H}_{F}' = \text{Span}[\mathcal{B}'_{F}]$. By noting $\rho_{ABF} \in \mathcal{D} \big( \mathcal{H}_{A} \otimes \mathcal{H}_{B} \otimes \mathcal{H}_{F}\big)$, $\rho_{AF}$ is given by
\begin{equation}
    \begin{aligned}
    \rho_{AF} &= \tr_{B}[\rho_{ABF}] \\&= \sum_{b_{H},b_{V}=0}^{\infty}(\id_{A} \otimes \bra{b_{H},b_{V}} \otimes \id_{F}) \rho_{ABF} (\id_{A} \otimes \ket{b_{H},b_{V}} \otimes \id_{F})\\
    &= \sum_{i,j = 0}^{3}\ketbra{i}{j}_{A} \otimes \sum_{s,s'=0}^{d-1}m_{i,j}^{s,s'}  \sum_{b_{H},b_{V}=0}^{\infty} (\bra{b_{H},b_{V}} \otimes \id_{F})\ketbra{s}{s'}(\ket{b_{H},b_{V}} \otimes \id_{F}).
    \end{aligned}
\end{equation}
Note that
\begin{equation}
    (\bra{b_{H},b_{V}} \otimes \id_{F})\ket{s} = \braket{b_{H}}{b_{H}^{s}}\cdot \braket{b_{V}}{b_{V}^{s}} \cdot \ket{f_{H}^{s},f_{V}^{s}} = \delta_{b_{H},b_{H}^{s}}\delta_{b_{V},b_{V}^{s}} \cdot \ket{f_{H}^{s},f_{V}^{s}}.
\end{equation}
We therefore see that inner summand is zero unless $b_{H} = b_{H}^{s} = b_{H}^{s'}$ and $b_{V} = b_{V}^{s} = b_{V}^{s'}$. We thus define
\begin{equation}
    F_{s,s'} = \sum_{b_{H},b_{V}=0}^{\infty} \delta_{b_{H},b_{H}^{s}} \cdot \delta_{b_{V},b_{V}^{s}} \cdot \delta_{b_{H},b_{H}^{s'}} \cdot \delta_{b_{V},b_{V}^{s'}} \cdot \ketbra{f_{H}^{s},f_{V}^{s}}{f_{H}^{s'},f_{V}^{s'}},
\end{equation}
allowing us to write
\begin{equation}
    \rho_{AF} = \sum_{i,j = 0}^{3}\ketbra{i}{j}_{A} \otimes \sum_{s,s'=0}^{d-1}m_{i,j}^{s,s'}F_{s,s'} . 
\end{equation}
Note that to compute $F_{s,s'}$ in practice, we only need to consider finitely many terms in the sum over $b_{H},b_{V}$. This is because the possible values of $b_{H}^{s}$ and $b_{V}^{s}$ are bounded from above due to the finite size of the subspace $\mathcal{H}_{BF}'$. 

\section{Bounds on the weight for the dimension reduction technique}\label{app:weight}
In this appendix, we show how an upper bound on the weight $W$ in \cref{eq:finalOptDR} can be derived from the frequency of double click events observed by Bob. Consider the active detection scenario with perfect-efficiency detectors. Recall that Bob's double-click POVM element in $H/V$ is 
\begin{equation}
\phv := \sum_{n_V=1}^\infty \sum_{n_H=1}^\infty \dyad{n_H,n_V},
\end{equation}
and we define $\pad$ for $D/A$ analogously. Note that $\phv$ is related to $N_{HV}$ defined in \cref{app:POVM} via $N_{HV} = \phv \otimes \id_{F}$. If $H/V$ is sampled a fraction $p_{z}$ of the time, the total double-click POVM is
\begin{equation}
\pdc := p_{z} \phv+ (1-p_{z}) \pad
\end{equation}
We denote the total $N$-photon subspace of $\mathcal{H}_{B}$, for $N=0,1,2,...$, by $\HN$. It is the support of the projector
\begin{align}
\pin& := \sum_{n=0}^N \dyad{n,N-n}_{H/V}\\
&=\sum_{n=0}^N \dyad{n,N-n}_{D/A},
\end{align}
and note that $P_{DC}$ is block-diagonal in the blocks formed by these projectors.

For states living entirely in the $N$-photon subspace ($N>0$), we now consider the smallest probability of observing a double-click. Formally, this amounts to solving
\begin{align}
\min \tr[\rho \pdc]\\
\rho \in \mathcal{D}\left({\HN}\right).
\end{align}
Some inspection reveals this is equivalent to finding the smallest eigenvalue $\lambda^N_{\text{min}}$ of the $N$-photon block of $\pdc$. Let's denote this block by $\pdcn=\pin \pdc \pin$.

\begin{lemma}
The smallest eigenvalue of $\pdcn$ is
\begin{equation}
\lambda^N_{\mathrm{min}}=\begin{cases}
\frac{1}{2} -\frac{1}{2} \sqrt{(2p_{z}-1)^2+8\cdot 2^{-N}(p_{z}-p_{z}^2)} & N \textrm{odd}\\
\frac{1}{2} -\frac{1}{2} \sqrt{(2p_{z}-1)^2+16\cdot 2^{-N}(p_{z}-p_{z}^2)} & N \textrm{even}
\end{cases}
\end{equation}
\label{lem:eig}
\end{lemma}

\begin{proof}
Observe $\phv^{N}=\pin - \dyad{0,N}_{H/V} - \dyad{N,0}_{H/V}$ and analogously $\pad^{N}=\pin - \dyad{0,N}_{D/A} - \dyad{N,0}_{D/A}$. This allows us to write, 
\begin{equation}
\pdc^{N}=\pin - p_{z}\dyad{0,N}_{H/V} - p_{z}\dyad{N,0}_{H/V}-(1-p_{z}) \dyad{N,0}_{D/A} - (1-p_{z}) \dyad{0,N}_{D/A}.
\end{equation}
Thus, the smallest eigenvalue of $\pdc$ is one minus the largest eigenvalue of $p_{z}\dyad{0,N}_{H/V} + p_{z}\dyad{N,0}_{H/V}+(1-p_{z}) \dyad{N,0}_{D/A} +(1-p_{z}) \dyad{0,N}_{D/A}$. The eigenvalues of this rank-4 operator can be computed via its Gram matrix. Namely, we use the fact that a state $\sum_{i} \lambda_i \dyad{\psi_i}$ has (up to padding with zero eigenvalues) the same eigenvalues as $\sum_{ij} \sqrt{\lambda_i \lambda_j} \braket{\psi_i}{\psi_j}\dyad{i}{j}$, where $\{\ket{i}\}_{i}$ forms an orthonormal basis but $\{\ket{\psi_i}\}_{i}$ may not. This follows by taking a purification of the original mixed state followed by the partial trace. The Gram matrix of the ensemble $\sqrt{p_{z}} \ket{N,0}_{H/V}$, $\sqrt{p_{z}} \ket{0,N}_{H/V}$, $\sqrt{1-p_{z}} \ket{N,0}_{D/A}$, $\sqrt{1-p_{z}} \ket{0,N}_{D/A}$ is
\begin{equation}
G=\begin{bmatrix}
p_{z} & 0 & 2^{-N/2}\sqrt{p_{z}(1-p_{z})} & 2^{-N/2}\sqrt{p_{z}(1-p_{z})}\\
0 & p_{z} & 2^{-N/2}\sqrt{p_{z}(1-p_{z})} & (-1)^N 2^{-N/2}\sqrt{p_{z}(1-p_{z})}\\
2^{-N/2}\sqrt{p_{z}(1-p_{z})} & 2^{-N/2}\sqrt{p_{z}(1-p_{z})} & 1-p_{z} &0\\
2^{-N/2}\sqrt{p_{z}(1-p_{z})} & (-1)^N 2^{-N/2}\sqrt{p_{z}(1-p_{z})} &0 & 1-p_{z}
\end{bmatrix}.
\end{equation}
This follows from the transformation of creation operators by beam splitters which gives $\ket{N,0}_{D/A} = 2^{-N/2} (\ket{N,0}_{H/V}+\ket{0,N}_{H/V} )+... $ and $\ket{0,N}_{D/A} = 2^{-N/2} (\ket{N,0}_{H/V}+(-1)^N\ket{0,N}_{H/V} )+... $, where all unspecified terms are of the form $\ket{a,N-a}$ with $0<a<N$. Symbolically computing the eigenvalues of $G$ and taking the largest completes the proof.
\end{proof}
\noindent We remark that $\lambda_{\text{min}}^1,\lambda_{\text{min}}^2=0$, $\lambda_{\text{min}}^{2N}=\lambda_{\text{min}}^{2N-1}$, and $\lambda_{\text{min}}^{N}\leq\lambda_{\text{min}}^{N+1}$.

We now show how \cref{lem:eig} translates to a lower bound on the weight inside a desired finite-dimensional subspace. Let the subspace of interest be that of up to $N$ photons, i.e., the support of the projector $\Pi^N= \sum_{k=0}^N \pik$.

\begin{corollary}
Let $\rho \in \mathcal{D}(\mathcal{H}_{B})$. If $\rho$ satisfies $\tr[\rho \pdc]=q$, then $\tr[\rho \Pi^N] \geq 1-\frac{q}{\lambda_{\mathrm{min}}^{N+1}}$. Note $N$ must be larger than 1 to obtain a nontrivial bound. \label{cor:weight}
\end{corollary}
\begin{proof}
By direct calculation,
\begin{align}
q&=\tr[\rho \pdc]\\
&= \tr[\rho \sum_{k=0}^\infty \pik \pdc]\\
&=\tr[\rho \sum_{k=0}^\infty \pik \pdc \pik]\\
&\geq \tr[\rho \sum_{k=N+1}^\infty \pik \pdc \pik]\\
&\geq \tr[\rho \sum_{k=N+1}^\infty \pik \lambda_{\text{min}}^k]\\
&\geq \tr[\rho \sum_{k=N+1}^\infty \pik \lambda_{\text{min}}^{N+1}]\\
&= \lambda_{\text{min}}^{N+1} \tr[\rho (\id - \Pi^N) ],
\end{align}
where the second equality follows by inserting identity, the third because $\pdc$ is block-diagonal, the fourth because we drop nonnegative terms, the fifth because $\lambda_{\text{min}}^k$ are the smallest eigenvalues of the blocks, and the final equality follows from the fact that the sequence $\lambda_{\text{min}}^{k}$ is increasing with $k$. Rearranging, we obtain
\begin{align}
\tr[\rho \Pi^N] \geq 1-\frac{q}{\lambda_{\text{min}}^{N+1}} 
\end{align}
as claimed.
\end{proof}

To conclude this appendix, we now show how \cref{cor:weight} can be used to obtain a bound on the weight $W$ in the bypass scenario, as discussed in \cref{sec:dr}, when $\eta_{T} = 1$. Here, both $B$ and $F$ have infinite-dimensional Hilbert spaces, and for a projector of the form $\Pi_B^{N} \otimes \Pi_F^{N}$, we can bound the weight as follows:
\begin{align}
\tr[\rho (\Pi_B^{N} \otimes \Pi_F^{N})] &= \tr[\rho (\Pi_B^{N} \otimes \id_F^{N})] - \tr[\rho (\Pi_B^{N} \otimes (\id_F-\Pi_F^{N}))]\\ 
&\geq \tr[\rho (\Pi_B^N \otimes \id_F)] - \tr[\rho (\id_B \otimes (\id_F-\Pi_F^N))]\\ 
&= 1-W_B - W_F
\end{align}
where $W_{B}$ and $W_{F}$ are the weights of $\rho$ outside the support of $\Pi_{B}^N$ and $\Pi_{F}^N$, respectively. Because $\eta_{T} = 1$, we can attribute all double click events to signals from mode $B$, and therefore $W_B$ can be bounded using the frequency of double-click events observed by Bob according to \cref{cor:weight}. Also note that $W_F$ can be calculated directly because the bypass channel is fully characterized via the source replacement constraints, i.e., we know the reduced state $\rho_{F}$. By linearity, this implies that for a sum of projectors $\sum_{k} \Pi_{B}^{k} \otimes \Pi_{F}^{k}$, where in each term the projectors on $B$ and $F$ can be different, we have $\tr[\rho (\sum_k \Pi^k_B \otimes \Pi^k_F)]\geq \sum_k (1-W_B^k - W_F^k)$.


%

\end{document}